\documentclass[12pt]{article}
\usepackage[utf8]{inputenc}
\usepackage{float}
\usepackage{amsfonts, amsmath, amssymb, amsthm, mathrsfs}
\usepackage{graphicx}
\usepackage{enumerate}
\usepackage{natbib}
\usepackage{url} 

\newcommand{\blind}{1}

\newtheorem{theorem}{Theorem}
\theoremstyle{plain}
\newtheorem{assumption}{Assumption} 
\newtheorem{lemma}{Lemma}

\newcommand{\nn}{\nonumber}

\begin{document}

\def\spacingset#1{\renewcommand{\baselinestretch}%
{#1}\small\normalsize} \spacingset{1}


\if1\blind
{
  \title{\bf Estimation of Causal Effects Under K-Nearest Neighbors Interference}
  \author{
  	Samirah Alzubaidi\\
    Department of Mathematics, Al-Qunfudah University College,\\ Umm Al-Qura University \\
    and \\
    Michael J. Higgins \\
    Department of Statistics, Kansas State University}
  \maketitle
} \fi

\if0\blind
{
  \bigskip
  \bigskip
  \bigskip
  \begin{center}
    {\LARGE\bf Estimation of Causal Effects Under K-Nearest Neighbors Interference}
\end{center}
  \medskip
} \fi



 
\bigskip
\begin{abstract}

Considerable recent work has focused on methods for analyzing experiments which exhibit \textit{treatment interference}---that is, when the treatment status of one unit may affect the response of another unit.
Such settings are common in experiments on social networks. 
We consider a model of treatment interference---the $K$-nearest neighbors interference model (KNNIM)---for which the response of one unit depends not only on the treatment status given to that unit, but also the treatment status of its $K$ ``closest'' neighbors.
We derive causal estimands under KNNIM in a way that allows us to identify how each of the $K$-nearest neighbors contributes to the indirect effect of treatment.
We propose unbiased estimators for these estimands and derive conservative variance estimates for these unbiased estimators.
We then consider extensions of these estimators under an assumption of no weak interaction between direct and indirect effects.
We perform a simulation study to determine the efficacy of these estimators under different treatment interference scenarios.
We apply our methodology to an experiment designed to assess the impact of a conflict-reducing program in middle schools in New Jersey, and we give evidence that the effect of treatment propagates primarily through a unit's closest connection.



\end{abstract}

\noindent%
\vfill

\newpage
\spacingset{1.9} 
\section{Introduction}
\label{sec:intro}


In randomized experiments, assessing causal effects requires special care when the treatment condition assigned to one unit is allowed to affect the response of other units. 
Under this setting, a 
unit's outcome is not only influenced by its own treatment status---a \textit{direct effect} of treatment---but may also be influenced by other units' treatments---an \textit{indirect effect}~\citep{sobel2006randomized, rosenbaum2007interference, hudgens2008toward}.
In causal inference terminology, these experiments exhibit 
\textit{treatment interference}, \textit{treatment spillover}, \textit{network effects}, or \textit{peer effects}.
 This interference is especially common in settings with a social factor where units are allowed to interact with each other, for example, in studies on social media networks.

Experiments exhibiting treatment interference violate 
the \textit{stable unit treatment value assumption} (SUTVA)---a foundational assumption of traditional causal inference methods~\citep{splawa1990application, rubin1980randomization, holland1986statistics, imbens2015causal}.
In particular, SUTVA requires that the treatment assigned to one unit affects only the outcome of that unit and does not affect the outcomes of other units.
The presence of interference may complicate statistical analysis and lead to inaccurate inference if not carefully taken into account~\citep{sobel2006randomized}.


Traditionally, the effect of social influence and interaction between units has been viewed as a nuisance prohibiting accurate estimation of the direct effect of treatment.
Considerable work has focused on designing experiments to mitigate the effect of treatment interference, for example, through clustering units that are likely to interact with each other and assigning treatment to these clusters instead of individual units~\citep{ugander2013graph, gui2015network, eckles2016design}. 
However, recent applications---for example, studies conducted on social media platforms and those evaluating the efficacy of vaccination strategies---have giving rise to studies in which quantifying and estimating interference effects is of primary interest~\citep{hudgens2008toward, aronow2017estimating, forastiere2020identification,  sussman2017elements, toulis2013estimation,  alzubaidi2022detecting}.

Methods for estimating indirect effects often begin by classifying the interaction through defining an exposure mapping on the units under study---a network where nodes represent units under study and edges between vertices indicate that the corresponding units may interact with each other.
There may additionally be an interaction measure computed between each pair of units in the exposure mapping indicating the strength of that interaction.
Common settings include partial interference, in which units are partitioned into groups and treatment interference can persist within groups but not across groups~\citep{sobel2006randomized, rosenbaum2007interference, hudgens2008toward,  tchetgen2012causal, basse2018analyzing}, and neighborhood interference, in which interference is not allowed to persist outside of a small neighborhood of a unit~\citep{sussman2017elements}.  
Once the nature of the treatment interference is specified, estimators can then be defined to estimate both direct and indirect effects~\citep{aronow2017estimating}.

We build on this literature by developing estimators for these 
treatment effects under the $K$-nearest neighbors interaction model (KNNIM)~\citep{alzubaidi2022detecting}.
In this model, the treatment given to one unit may interfere with response another unit if the first unit is one of the $K$ ``closest'' individuals to the second unit with respect to a given interaction measure.
Quantifying $K$-nearest neighbors effects may help researchers tease out peer effects induced by, for example, interactions between best friends, spouses, siblings, and/or close colleagues.
Additionally, this model has several appealing properties.
First, this model allows for users with stronger interactions to produce larger indirect effects than users with weaker interactions, and will ignore potential indirect effects due to dilapidated, but technically present, connections (e.g.~Facebook friends that no longer interact with each other).  
Second, the marginal and joint probabilities for possible treatment exposures have closed-form expressions under common experimental settings, allowing for unbiased estimation of treatment effects and precise estimation of standard errors.
Finally, KNNIM may be effective in estimating indirect effects in the presence of non-transitive \textit{influencer} effects---effects induced by users that have influence over a large number of individuals, but may themselves only directly interact with a handful of individuals.


In this paper, using a potential outcomes approach, we define causal estimands for direct and $K$-nearest-neighbor indirect effects.  
We then derive Horvitz-Thompson estimators \citep{horvitz1952generalization} for these estimands that are unbiased given exact marginal and joint probabilities for possible treatment exposures, and derive conservative standard errors for these estimators. 
We provide a closed-form solution to compute these exposure probabilities under completely-randomized and Bernoulli-randomized experimental designs.
We then demonstrate how these estimators may have significantly stronger precision under an assumption of no interaction between direct and indirect effects.
We conclude by showcasing the efficacy of these methods via simulation and application to a field experiment designed to reduce conflict among middle school students in New Jersey~\citep{paluck2016changing, aronow2017estimating}.


The paper is organized as follows.
Section \ref{makereference3.2} sets up the notation and preliminaries.
The $K$-nearest neighbors interference model is provided in Section \ref{makereference3.3}.
Section \ref{makereference3.4} defines causal effects under KNNIM.
Proposed unbiased estimators with the derived properties are given in Section \ref{makereference3.5}. 
Section \ref{makereference3.6} provides estimators with stronger precision under an assumption of no interaction between direct and indirect effects.
Section \ref{makereference3.7} presents variance estimation.
Simulation studies and real data analysis to demonstrate the proposed methods are provided in Sections \ref{makereference3.8}, \ref{makereference3.9}, and \ref{makereference3.10}. We conclude in Section \ref{makereference3.11}.

\subsection{Motivating Example: Reducing Conflict in Schools}
\label{sec:motivexamp}

To motivate our approach, we refer to a recent randomized field experiment to assess the efficacy of an anti-conflict program on middle schools in New Jersey~\citep{paluck2016changing}. 
The experiment was explicitly designed to determine whether benefits of the program can be propagated through social interactions between students.
In particular, the program was implemented through randomly selecting ``seed students'' from a pre-selected list of eligible seeds within all participating schools.
Seed students were asked to actively participate in and advocate for the anti-conflict program.
School-wide benefits of the program were then propagated through interactions with these seed students.
Analysis was performed only on students that were eligible to be seeds ($N =$ 2,451).
Overall, the program showed a statistically significant improvement in anti-conflict behaviors among students. 
For details of the randomization and analysis, please refer to~\citet{paluck2016changing} and~\citet{aronow2017estimating}.

Notably, to assess potential pathways for treatment interference, students were asked to identify, in order, the 10 other students that they spent the most time with during the previous few weeks.  
This yields a unique dataset in which the strength of the interaction between two individuals under study is explicitly recorded.
Hence, an interference model that allows for direct incorporation of these relative strengths of the interactions, such as KNNIM, can yield additional insights into the structure of the treatment interference.
For example, in Section~\ref{makereference3.10}
, we give evidence that the indirect effect of the treatment was propagated primarily through the student with the ``strongest connection'' to the seed student.

\section{Notation and Preliminaries}
\label{makereference3.2}

Consider an experiment on $N$ units where each unit is 
assigned either a treatment status or a control status. 
The Neyman-Rubin Causal Model (NRCM)~\citep{splawa1990application, rubin1974estimating, holland1986statistics} is a commonly-assumed model of response for causal inference.  
Under this model, the observed response of a unit is determined by the treatment status given to that unit and the \textit{potential outcomes} for that unit---the hypothetical responses of that unit under the possible treatment statuses.
A fundamental assumption of this model is the stable-unit treatment value assumption (SUTVA), which requires that there is only a single version of each treatment status and the response of a unit is unaffected by the treatment status of any other unit~\citep{splawa1990application, rubin1980randomization, holland1986statistics, imbens2015causal}. 
Of note, SUTVA is violated for experiments that exhibit treatment interference~\citep{cox1958planning,rubin1980randomization}.  
Failing to account for violations of SUTVA can lead to inaccurate treatment effect estimates~\citep{sobel2006randomized}.

Under interference, the effect of a treatment on a unit may occur through direct application of the treatment to that unit, indirectly through application of treatment to units that affect the response of the original unit, or both~\citep{hudgens2008toward}.
When interference is allowed to take completely arbitrary forms, treatment effect estimates are often estimated with very low power, or may be unidentifiable~\citep{aronow2017estimating}.
Thus, to make progress on 
this problem, 
researchers often make assumptions restricting how interference can propagate across units~\citep{toulis2013estimation, aronow2017estimating, ugander2013graph, sussman2017elements}.

To elucidate these ideas,
we extend the potential outcomes framework to account for both direct and indirect treatment components.  
Let $W_i$ be a treatment indicator for unit $i$---that is, $W_i = 1$ if unit $i$ is given treatment and $W_i = 0$ otherwise.
Let $y_{i}(\mathbf W) = y_{i}(W_i, \mathbf W_{-i})$ denote the potential outcome of unit $i$ under treatment allocation $\mathbf W = (W_1, W_2, \ldots, W_N) \in \{0,1\}^N$, where unit $i$ is given treatment $W_i$, and the remaining treatment statuses are allocated according to $\mathbf W_{-i}$.
Responses $Y_i$ satisfy
\begin{equation}
    Y_i = \sum_{\mathbf W \in \{0,1\}^N} y_{i}(\mathbf W)\mathbf{1}(\mathbf W'=\mathbf W),
\end{equation}
where $\mathbf{1}(\mathbf W'=\mathbf W)$ is an indicator variable that is equal to 1 if and only if the observed treatment status $\mathbf W' = \mathbf W$.
That is, the response of unit $i$ only depends on the potential outcomes of unit $i$ and the observed treatment allocation across all $N$ units.


Without making assumptions on the amount of interference allowed in a study, it may be impossible to estimate common causal quantities of interest in any practical way---for example, each unit may have up to $2^N$ potential outcomes when interference is unconstrained.
Thus, researchers often place strong restrictions on the nature of the treatment interference~\citep{toulis2013estimation, aronow2017estimating, ugander2013graph, sussman2017elements, alzubaidi2022detecting}. 
This often begins by constructing an \textit{exposure mapping} $G = (V, E)$---a directed graph where each vertex $i \in V$ represents a unit under study and an edge $\vec{ij} \in E$ indicates that the treatment status of unit $i$ may potentially interfere with the response of unit $j$.  
Each edge $\vec{ij} \in E$ may also have a weight $d(i,j)$ denoting the strength of the the potential interference which may be observed through studying interactions between $i$ and $j$.  
Stronger interference between $i$ and $j$ may either correspond to  a larger or smaller value of $d(i,j)$ depending on the context of the problem---for this paper, stronger interactions are associated with smaller values of $d(i,j)$.  
Under an assumed exposure mapping $G$, for any two treatment allocations $\mathbf W$, $\mathbf W'$, we have that $y_i(\mathbf W) = y_i(\mathbf W')$ if $W_i = W_i'$ and $W_j = W_j'$ for all $j \in V$ such that $\vec{ji} \in E$;
if $\vec{ki} \notin E$, then treatment statuses $W_k$, $W'_k$ may differ without affecting equality of the potential outcomes.

Once the exposure mapping is specified, models may further restrict the nature of interference.
For example, a common assumption is that interference can only occur if a certain number or fraction of neighbors within the exposure mapping are given the treatment condition.
However, few existing models specify exactly which neighbors in the exposure mapping are allowed to interfere with a unit's response, or allow for indirect effects to differ across neighbors. 
As an alternative, we consider a model where interference of treatment on a unit $i$ is restricted to its $K$–nearest neighbors~\citep{ alzubaidi2022detecting}. 
This model allows neighbors with stronger interactions to contribute larger indirect effects, and limits the ability of weakly-interacting units to affect response.

\section{$K$-Nearest Neighbors Interference Model}
\label{makereference3.3}


The $K$-nearest neighbors interference model (KNNIM) is a recently-proposed model in which a unit $j$ is only allowed to interfere with the response of unit $i$ if $j$ is within $i$'s $K$-neighborhood~\citep{ alzubaidi2022detecting}.
The \textit{$K$-neighborhood} of unit $i$, denoted $\mathcal{N}_{iK}$, is the set of the $K$ ``closest'' units to unit $i$:
\begin{align*}
    \mathbf{\mathcal{N}}_{iK} = \{j : d(i,j) \leq d(i, (K)), j = 1,2, \ldots ,K, ~ j \ne i\}.
\end{align*}
where $d(i, (K))$ denotes the $K^\text{th}$ smallest  value of $(d(i,j))^N_{j=1}$.   
Ties between $d(i,j)$ may be broken arbitrarily to ensure that $|\mathbf{\mathcal{N}}_{iK}| = K$.
Define $\mathcal{N}_{-iK} = V \setminus (i \cup N_{iK} )$ as all units in $V$ that are outside of $i$'s $K$-neighborhood.
Note that the sets
$\{i\}, \{\mathcal{N}_{ik}\}, \{\mathcal{N}_{-ik}\}$ form a partition of $V$.

Let 
 $\mathbf{W} = (W_{i}, \mathbf{W}_{\mathcal{N}_{iK}}, \mathbf{W}_{\mathcal{N}_{-iK}})$ denote treatment assignment vector for all units $N$, partitioned into the treatment given to unit $i$, the treatments given to $i$'s $K$-nearest neighbors, and the treatments given to all other units.
For clarity, treatment statuses in $\mathbf{W}_{\mathcal{N}_{iK}}$ and  $\mathbf{W}_{\mathcal{N}_{-iK}}$ are given in increasing order with respect to $d(i,j)$ 
(\textit{e.g.}~the first entry of $ \mathbf{W}_{\mathcal{N}_{iK}}$ is the treatment assignment given to the nearest neighbor of unit $i$).
The defining assumption of KNNIM is as follows:  
\begin{assumption}
\label{Assumption1}
    ($K$-Neighborhood Interference Assumption (K-NIA)). 
The potential outcomes $y_i(\mathbf W)$ for all units $i$ satisfy
\begin{equation} 
y_{i}(W_{i}, \mathbf{W}_{\mathcal{N}_{iK}}, \mathbf{W}_{\mathcal{N}_{-iK}}) = y_{i}(W_{i}, \mathbf{W}_{\mathcal{N}_{iK}}, \mathbf{W}'_{\mathcal{N}_{-iK}}).
\end{equation}
\end{assumption}
In words, K-NIA ensures that the potential outcome of unit $i$ is only affected by its own treatment status and the treatment statuses of its $K$-nearest neighbors.
The treatment statuses of units outside the $K$-neighborhood will not affect the potential outcome of unit $i$.
Thus, this model restricts the number of potential outcomes to be $2^{K+1}$ for each unit. 
For brevity, we will suppress the treatment assignment outside of the $K$-nearest neighbors when denoting potential outcomes under KNNIM:
$y_{i}(W_{i}, \mathbf{W}_{\mathcal{N}_{iK}}) = y_{i}(W_{i}, \mathbf{W}_{\mathcal{N}_{iK}}, \mathbf{W}_{\mathcal{N}_{-iK}})$.

The choice of $K$ is ultimately left to the researcher, though large values of $K$ may not be sufficiently restrictive to allow for reliable estimates and inferences.
A detailed description on how to choose $K$ is given in~\citet{alzubaidi2022detecting}.
As a crude, but useful, rule-of-thumb, we find it practical to choose $K$ such that each exposure used in the estimation of treatment effects is observed at least 30 times.

\section{Causal Estimands under KNNIM}
\label{makereference3.4}

Using the potential outcomes framework and following \citet{hudgens2008toward}, we now define causal estimands under KNNIM. 
We start with general definitions of direct and indirect effects and conclude with KNNIM-specific nearest neighbors effects.

\subsection{Direct, Indirect, and Total Effects}
\label{section4.1}
    
The \textit{average direct effect} (ADE) $\delta_{dir}$ is the average difference in a unit's potential outcomes when changing that unit's treatment status and holding all other units' treatment status fixed.
It may be defined as
\begin{equation}
\label{eq:direff}
\delta_{dir} = \frac{1}{N} \sum_{i = 1}^N (y_i(1,\mathbf{1}) - y_i(0,\mathbf{1})),
\end{equation}
where $\mathbf{1}$ denotes a vector of 1's of length $K$.
In contrast to direct effect, the \textit{average indirect effect} (AIE) $\delta_{ind}$ is defined as the average difference in a unit's potential outcome when changing all other treatment statuses from control to treated, holding its own treatment fixed.
It may be defined as
\begin{equation}
\label{eq:indeff}
\delta_{ind} = \frac{1}{N} \sum_{i = 1}^N (y_i(0,\mathbf{1}) - y_i(0,\mathbf{0})),
\end{equation}
where $\mathbf{0}$ denotes a vector of 0's of length $K$.
The \textit{average total effect} (ATOT) $\delta_{tot}$ measures the average difference in potential outcomes between all units receiving treatment and all units receiving control:
\begin{equation}
\delta_{tot} = \frac{1}{N} \sum_{i = 1}^N (y_i(1,\mathbf{1}) - y_i(0,\mathbf{0})). 
\end{equation}
Note that these quantities are defined to satisfy
\begin{equation}
    \label{eq:trtdecomp}
    \delta_{tot} = \delta_{dir} + \delta_{ind}.
\end{equation}
When SUTVA holds, $\delta_{tot} = \delta_{dir}$ and $\delta_{ind} = 0$.

Note, we may define $\delta_{dir} = N^{-1}\sum_{i=1}^N (y_i(1, \mathbf{0}) - y_i(0, \mathbf{0}))$ and $\delta_{ind} = N^{-1}\sum_{i=1}^N (y_i(1, \mathbf{1}) - y_i(1, \mathbf{0}))$ while still ensuring that~\eqref{eq:trtdecomp} holds.  
These quantities may differ from~\eqref{eq:direff} and~\eqref{eq:indeff} if there is interaction between direct effects and indirect effects---that is, if the differences $y_i(1,\mathbf W_{-i}) - y_i(0,\mathbf W_{-i})$ differ depending on the allocation of treatment given to $\mathbf W_{-i}$~\citep{alzubaidi2022detecting}.


\subsection{The $\ell^{th}$--Nearest Neighbor Indirect Effect}

We now define nearest-neighbor average treatment effects.
Under KNNIM, these estimands are of primary researcher interest.
Let $\mathbf{W}^{*}_{\ell} = (W^*_{\ell, 1}, W^*_{\ell, 2}, \ldots, W^*_{\ell, K})
\in\{0,1\}^K$ 
denote the treatment assignment vector of length $K$ where the first $\ell$ nearest neighbors are given treatment and the rest are given control:
\begin{equation}
    \mathbf{W}^*_{\ell, j} = \left\{ 
    \begin{array}{ll}
        1, & j \leq \ell,\\
        0, &\text{otherwise}.
    \end{array}\right.
\end{equation}
Note that $\mathbf{W}^{*}_{K} = \mathbf 1$ and $\mathbf{W}^{*}_{0} = \mathbf{0}$. 
Following our definitions from Section~\ref{section4.1}, we define the \textit{average $\ell^{th}$--nearest neighbor indirect effect} (A$\ell$NNIE) as
\begin{equation}
    \delta_{\ell} = \frac{1}{N} \sum_{i = 1}^N (y_i(0,\mathbf{W}^{*}_{\ell}) - y_i(0,\mathbf{W}^{*}_{\ell-1} )).
\end{equation}
Note that $\mathbf{W}^{*}_{\ell}$ and $\mathbf{W}^{*}_{\ell-1}$ are identical except 
that $W^*_{\ell, \ell} = 1$ and $W^*_{\ell-1, \ell} = 0$.
Hence, $\delta_{\ell}$ may be interpreted as the average difference in response due to the treatment status of the $\ell^{th}$--nearest-neighbor.
Additionally, under KNNIM, the AIE is the sum of the A$\ell$NNIEs.
\begin{equation} 
\delta_{ind} = \sum_{\ell = 1}^{K} \delta_{\ell}.
\end{equation}
To see this, note that $\sum_{\ell = 1}^{K} \delta_{\ell}$ is a telescoping sum that simplifies to
\begin{equation}
    \sum_{\ell = 1}^{K} \delta_{\ell} =  \sum_{\ell = 1}^{K}\frac{1}{N} \left(\sum_{i = 1}^N y_i(0,\mathbf{W}^{*}_{\ell}) - \sum_{i=1}^Ny_i(0,\mathbf{W}^{*}_{\ell-1} )\right) = \frac{1}{N} \sum_{i = 1}^N (y_i(0,\mathbf{1}) - y_i(0,\mathbf{0} )) = \delta_{ind}.
\end{equation}

\section{Horvitz–Thompson Estimators}
\label{makereference3.5}

We now derive Horvitz–Thompson (HT) estimators for the estimands described in Section~\ref{makereference3.4}. 
Our approach closely follows that in~\citet{aronow2017estimating}.
Of particular note, these HT estimators require computing, for each unit $i$, various marginal and joint probabilities related to the treatment allocation on $i\cup \mathcal N_{iK}$.
Thankfully, for many common designs, these probabilities can be computed exactly under KNNIM, and we give closed-form solutions for these probabilities under completely-randomized and Bernoulli-randomized designs.
When these probabilities cannot be computed exactly, they may still be estimated, for example, using the approach in~\citet{aronow2017estimating}.

Let $\pi_i(W,\mathbf{W}_{\mathcal{N}_{K}})$ denote the marginal probability that unit $i$ is given exposure $(W,\mathbf{W}_{\mathcal{N}_{K}})$---that is, the overall treatment allocation assigns treatment $W$ to unit $i$ and assigns treatment conditions $\mathbf{W}_{\mathcal{N}_{K}}$ to $i$'s $K$-neighborhood. 
Define $\pi_{ij}(W,\mathbf{W}_{\mathcal{N}_{K}})$ as the joint probability that units $i$ and $j$ are both given exposure $(W,\mathbf{W}_{\mathcal{N}_{K}})$ and define $\pi_{ij}((W,\mathbf{W}_{\mathcal{N}_{K}}), (W',\mathbf{W}'_{\mathcal{N}_{K}}))$ as the joint probability  that unit $i$ receives exposure $(W,\mathbf{W}_{\mathcal{N}_{K}})$ and unit $j$ receives exposure $(W',\mathbf{W}'_{\mathcal{N}_{K}})$.
Define indicator variables $I_{i}(W,\mathbf{W}_{\mathcal{N}_{K}})$ that are equal to 1 if unit $i$ is given exposure $(W,\mathbf{W}_{\mathcal{N}_{K}})$ and are 0 otherwise.
Finally, let 
\begin{equation}
    \bar y(W,\mathbf{W}_{\mathcal{N}_{K}}) = \frac{1}{N} \sum_{i=1}^N y_i(W,\mathbf{W}_{\mathcal{N}_{K}})
\end{equation}
denote the average potential outcome across all $N$ units for treatment allocation $(W,\mathbf{W}_{\mathcal{N}_{K}})$.
The Horvitz-Thompson (HT) estimator~\citep{horvitz1952generalization} for $ \bar y(W,\mathbf{W}_{\mathcal{N}_{K}})$ is
\begin{equation} 
 \bar{Y}_{HT}^{obs}(W,\mathbf{W}_{\mathcal{N}_{K}}) =  \frac{1}{N}\sum_{i = 1}^{N}  I_{i}(W,\mathbf{W}_{\mathcal{N}_{K}})\frac{Y_i}{\pi_{i}(W,\mathbf{W}_{\mathcal{N}_{K}})}.
\end{equation}
This estimator is unbiased for $\bar y(W,\mathbf{W}_{\mathcal{N}_{K}})$,
\begin{equation} 
\label{EXP}
\mathbf{E}(\bar{Y}_{HT}^{obs}(W,\mathbf{W}_{\mathcal{N}_{K}})) 
 =  \bar{y}(W,\mathbf{W}_{\mathcal{N}_{K}}),
\end{equation}
and has variance 
\begin{multline}\label{VAR}
\mathbf{Var}(\bar{Y}_{HT}^{obs}(W,\mathbf{W}_{\mathcal{N}_{K}})) 
 =  \frac{1}{N^2}\sum_{i = 1}^{N}\pi_{i}(W,\mathbf{W}_{\mathcal{N}_{K}})[1 - \pi_{i}(W,\mathbf{W}_{\mathcal{N}_{K}})]\left[ \frac{y_{i}(W,\mathbf{W}_{\mathcal{N}_{K}})}{\pi_{i}(W,\mathbf{W}_{\mathcal{N}_{K}})}\right]^2\\
  + \frac{1}{N^2}\sum_{i = 1}^{N}\sum_{j \ne i}[\pi_{ij}(W,\mathbf{W}_{\mathcal{N}_{K}}) - \pi_{i}(W,\mathbf{W}_{\mathcal{N}_{K}})\pi_{j}(W,\mathbf{W}_{\mathcal{N}_{K}})]\frac{y_{i}(W,\mathbf{W}_{\mathcal{N}_{K}})}{\pi_{i}(W,\mathbf{W}_{\mathcal{N}_{K}})}\frac{y_{j}(W,\mathbf{W}_{\mathcal{N}_{K}})}{\pi_{j}(W,\mathbf{W}_{\mathcal{N}_{K}})}.
\end{multline}
The covariance of the HT estimators under any two exposures  $(W,\mathbf{W}_{\mathcal{N}_{K}})$ and $(W',\mathbf{W}'_{\mathcal{N}_{K}})$ is 
\begin{multline}\label{COV}
\mathbf{Cov}(\bar{Y}_{HT}^{obs}(W,\mathbf{W}_{\mathcal{N}_{K}}),\bar{Y}_{HT}^{obs}(W',\mathbf{W}'_{\mathcal{N}_{K}})) =
 \\\frac{1}{N^2}\left(\sum_{i = 1}^{N}\sum_{j \ne i}\left[\pi_{ij}((W,\mathbf{W}_{\mathcal{N}_{K}}),(W',\mathbf{W}'_{\mathcal{N}_{K}})) -
 \pi_{i}(W,\mathbf{W}_{\mathcal{N}_{K}})\pi_{j}(W',\mathbf{W}'_{\mathcal{N}_{K}})\right] \right.
 \\ \times
\left. \frac{y_{i}(W,\mathbf{W}_{\mathcal{N}_{K}})}{\pi_{i}(W,\mathbf{W}_{\mathcal{N}_{K}})}
 \frac{y_{j}(W',\mathbf{W}'_{\mathcal{N}_{K}})}{\pi_{j}(W',\mathbf{W}'_{\mathcal{N}_{K}})} \right)
 - \frac{1}{N^2}\sum_{i = 1}^{N}  y_{i}(W,\mathbf{W}_{\mathcal{N}_{K}})y_{i}(W',\mathbf{W}'_{\mathcal{N}_{K}})
.
\end{multline}

From~\eqref{EXP}, \eqref{VAR}, and \eqref{COV}, the expectation and variance for the difference in HT estimators for the average response under any two unique exposures $(W,\mathbf{W}_{\mathcal{N}_{K}})$, $(W',\mathbf{W}'_{\mathcal{N}_{K}})$ can be computed as follows:
\begin{theorem}\label{Theorem1}
\begin{equation} 
\mathbf{E}(\bar{Y}_{HT}^{obs}(W,\mathbf{W}_{\mathcal{N}_{K}}) - \bar{Y}_{HT}^{obs}(W',\mathbf{W'}_{\mathcal{N}_{K}})) =   \bar y(W,\mathbf{W}_{\mathcal{N}_{K}}) -   \bar y(W',\mathbf{W'}_{\mathcal{N}_{K}}),
\end{equation}
\begin{multline}
\label{VARDIF}
\mathbf{Var}(\bar{Y}_{HT}^{obs}(W,\mathbf{W}_{\mathcal{N}_{K}}) - \bar{Y}_{HT}^{obs}(W',\mathbf{W'}_{\mathcal{N}_{K}}))=\\
\mathbf{Var}(\bar{Y}_{HT}^{obs}(W,\mathbf{W}_{\mathcal{N}_{K}})) + \mathbf{Var}(\bar{Y}_{HT}^{obs}(W',\mathbf{W}'_{\mathcal{N}_{K}})) 
-2\mathbf{Cov}(\bar{Y}_{HT}^{obs}(W,\mathbf{W}_{\mathcal{N}_{K}}),\bar{Y}_{HT}^{obs}(W',\mathbf{W}'_{\mathcal{N}_{K}})).
\end{multline}
\end{theorem}

We can then obtain unbiased estimators for the ADE, AIE, ATOT, and the A$\ell$NNIE as follows:
\begin{align} 
\widehat{\delta}_{{HT,dir}} &= \bar{Y}_{HT}^{obs}(1,\mathbf{1}) - \bar{Y}_{HT}^{obs}(0,\mathbf{1}),\\
\widehat{\delta}_{{HT,ind}} &= \bar{Y}_{HT}^{obs}(0,\mathbf{1}) - \bar{Y}_{HT}^{obs}(0,\mathbf{0}), \\
\widehat{\delta}_{{HT,tot}}  &= \bar{Y}_{HT}^{obs}(1,\mathbf{1}) - \bar{Y}_{HT}^{obs}(0,\mathbf{0}),\\
\widehat{\delta}_{HT,\ell} &= \bar{Y}_{HT}^{obs}(0,\mathbf W^{*}_{\ell}) - \bar{Y}_{HT}^{obs}(0,\mathbf W^{*}_{\ell-1}).
\end{align} 
Variances for these estimators are derived as in~\eqref{VARDIF}.
As with the estimands, we have the following relationships between these estimators:
\begin{align}
\label{sumofdirectindirectCH3.1}
\widehat{\delta}_{HT,tot} &= \widehat{\delta}_{HT,dir} + \widehat{\delta}_{HT,ind}, \\
\label{sumofdirectindirectCH3.2}
\widehat{\delta}_{HT,ind} &= \sum_{\ell = 1}^K \widehat{\delta}_{HT,\ell} .
\end{align}
\subsection{Marginal and Joint Exposure Probabilities \label{probcompute}}

One significant benefit of KNNIM is that this model allows for closed-form expressions 
of the marginal and joint exposure probabilities for many common experimental designs.
We now provide these exposure probabilities under completely-randomized and Bernoulli-randomized designs. 

We first introduce some notation to facilitate this discussion.
For any two units $i$ and $j$, let $b_{ij}$ denote the number of units shared by both $i$ and $j$'s closed $K$-neighborhoods:
$b_{ij} =  |(i \cup \mathcal N_{iK})\cap (j \cup \mathcal N_{jK})|$.
Consider two (not necessarily unique) treatment exposures $(W, \mathbf W_{\mathcal{N}_{K}})$, $(W', \mathbf W'_{\mathcal{N}_{K}})$ given to the closed $K$-neighborhood of units $i$ and $j$ respectively.
These exposures are called \textit{compatible} if they can co-occur within a given treatment assignment.
For example, if $j$ is unit $i$'s nearest neighbor, the exposures $(1,\mathbf 1)$ for unit $i$ and $(0, \mathbf 0)$ for unit $j$ are \textit{incompatible} since unit $j$ is given treatment in the first exposure and control in the second exposure---the probability that $i$ and $j$ can jointly observe two incompatible treatment exposures is 0.
For the exposure $(W, \mathbf W_{\mathcal{N}_{K}})$, let $N_{itK}$ and $N_{icK}$ denote the number of treated and control units respectively within $i \cup \mathcal N_{iK}$. 
For exposure $(W', \mathbf W'_{\mathcal{N}_{K}})$, let $N_{jitK}$ and $N_{jicK}$ denote the number of treated and control units in $(j \cup \mathcal N_{jK}) \setminus (i \cup \mathcal N_{iK}$)---that is, the number of treated and control units for units not already included in $i$'s closed neighborhood.
The specific units given treatment or control do not factor into the probability computations.



\subsubsection{Exposure Probabilities Under Complete Randomization}

In a completely-randomized design, the number of treated units $N_t$ in the study is selected prior to randomization.  
Each possible treatment has the same $\binom{N}{N_t}^{-1}$ probability of occurring.

The marginal probability that  $i \cup \mathcal N_{iK}$ receives exposure $(W, \mathbf W_{\mathcal{N}_{K}})$ is
\begin{equation}\label{CRDpi} 
\pi_{i}(W, \mathbf W_{\mathcal{N}_{K}}) = \frac{\binom {N-K-1} {N_{t}-N_{itK}}}{ \binom N{N_{t}}}.
\end{equation}
The joint probability that unit $i$ receives exposure $(W, \mathbf W_{\mathcal{N}_{K}})$ and unit $j$ receives exposure  $(W', \mathbf W'_{\mathcal{N}_{K}})$ is
\begin{multline}\label{JointCRD}
\pi_{ij}((W, \mathbf W_{\mathcal{N}_{K}}), (W', \mathbf W'_{\mathcal{N}_{K}})) =\\
\begin{cases}
      \frac{\binom {N-K-1} {N_{t}-N_{tK}}}{\binom N{N_{t}}} \frac{\binom {N-2K-2+b_{ij}} {N_{t}-N_{itK}-N_{jitK}}}{\binom {N-K-1} {N_{t}-N_{itK}}}, & (W, \mathbf W_{\mathcal{N}_{K}}),(W', \mathbf W'_{\mathcal{N}_{K}})~\text{are compatible for } i \text{ and } j, \\
       0, & \text{otherwise}.
    \end{cases}
\end{multline}

\subsubsection{Exposure Probabilities Under Bernoulli Randomization}

In a Bernoulli-randomized design, each unit has a pre-specified probability $p$ of being assigned treatment, and treatments are assigned independently across units (\textit{e.g.}~treatment assignment is determined for each unit by flipping a coin that has probability $p$ of landing heads).
Under Bernoulli-randomization, the marginal probability that  $i \cup \mathcal N_{iK}$ receives exposure $(W, \mathbf W_{\mathcal{N}_{K}})$ is
\begin{equation}\label{BRpi}
    \pi_{i}(W, \mathbf W_{\mathcal{N}_{K}}) = p^{N_{itK}}(1-p)^{N_{icK}},
\end{equation}
and joint probability of units $i$ and $j$ being exposed to  $(W, \mathbf W_{\mathcal{N}_{K}})$ and $(W', \mathbf W'_{\mathcal{N}_{K}})$ respectively is
\begin{multline}\label{JointBR}
\pi_{ij}((W_{i},W_{\mathcal{N}_{ik}}),(W'_{j},W'_{\mathcal{N}_{jk}})) 
 = \\
\begin{cases}
      p^{N_{itK}}(1-p)^{N_{icK}}p^{N_{jitK}}(1-p)^{N_{jicK}},  & (W, \mathbf W_{\mathcal{N}_{K}}),(W', \mathbf W'_{\mathcal{N}_{K}})~\text{are compatible for } i \text{ and } j, \\
       0, & \text{otherwise}.
    \end{cases}
\end{multline}

\section{Estimation Under No Weak Interaction Between Direct and Indirect Effects}
\label{makereference3.6}

Next, we consider an additional assumption on the responses of potential outcomes---\textit{no weak interaction between direct and indirect effects}.  
When it holds, this assumption allows for significantly more powerful estimators of direct and indirect effects and makes the two competing definitions of direct and indirect effects in Section~\ref{section4.1} equivalent.
Note, this assumption is similar to the additivity of main effects (AME) assumption considered by~\citet{sussman2017elements}, however, we strictly weaken this assumption by only requiring it to hold for the average of the potential outcomes.

\begin{assumption}\label{Assumption2}
(No Weak Interaction Between Direct and Indirect Effects) 
There is no weak interaction between direct and indirect effects if, for any two potential treatment exposures $\mathbf{W}_{\mathcal{N}_{K}}, \mathbf{W}'_{\mathcal{N}_{K}}$ on the $K$-neighborhoods
\begin{equation}
\label{noweakint}
\bar y(1,\mathbf{W}_{\mathcal{N}_{K}}) - \bar y (1,\mathbf{W}'_{\mathcal{N}_{K}}) -
(\bar y(0,\mathbf{W}_{\mathcal{N}_{K}}) - \bar y(0,\mathbf{W}'_{\mathcal{N}_{K}})) = 0.
\end{equation}
\end{assumption}

While strong, Assumption~\ref{Assumption2} may be plausible under certain settings. 
One such setting could occur if direct and indirect effects of treatment are multiplicative---i.e.~treatment effects may increase response by a certain percentage.  
In this case, taking the log of response may lead to model in which the stronger AME assumption holds.
For example, under KNNIM, we may consider the following model of response:
\begin{equation*}
        Y_i  = Y(0, \mathbf 0)\beta_t^{W_i}\prod_{\ell=1}^K \beta_{\ell}^{W_{i\ell}},
\end{equation*}
where $W_{i\ell}$ indicates whether the $\ell$th nearest neighbor of unit $i$ receives treatment.
After taking logs of both sides, we obtain
\begin{equation*}
        \log Y_i  = Y(0, \mathbf 0) + W_i\log \beta_t + \sum_{\ell=1}^K W_{i\ell} \log \beta_{\ell},
\end{equation*}
thereby satisfying AME---the direct effect is always $\log(\beta_t)$ regardless of the treatment status of any other unit.

Assumption~\ref{Assumption2} may also hold when direct treatment effects are weak relative to the indirect effects. 
For example, consider an experiment on a social network designed to assess the efficacy of an advertisement for an event.  
Treatment may involve posting a flyer on a user's ``wall'' if that user had previously expressed interest in attending that event, and response would indicate whether that user attends the event.
Since the user is already interested in the event, the user may not be particularly affected by the flyer on their wall, suggesting a weak direct treatment effect.  
However, if the user notices that their close friends also have flyers on their walls, they all may decide to attend the event together, suggesting a strong indirect effect.


To demonstrate the power of Assumption~\ref{Assumption2}, consider the case of estimating the AIE $\delta_{ind}$.  
When Assumption~\ref{Assumption2} holds,  both  $\bar{Y}_{HT}^{obs}(1,\mathbf{1}) - \bar{Y}_{HT}^{obs}(1,\mathbf{0})$ and $\bar{Y}_{HT}^{obs}(0,\mathbf{1}) - \bar{Y}_{HT}^{obs}(0,\mathbf{0})$ unbiasedly estimate $\delta_{ind}$---the latter quantity unbiasedly estimates $\bar{y}(0,\mathbf{1}) - \bar{y}(0,\mathbf{0})$, and under  Assumption~\ref{Assumption2},  $\bar{y}(0,\mathbf{1}) - \bar{y}(0,\mathbf{0}) = \bar{y}(1,\mathbf{1}) - \bar{y}(1,\mathbf{0}) = \delta_{ind}$. 
Thus, these two estimators can be combined to yield a more powerful unbiased estimate of the AIE:
\begin{equation}
    \hat\delta_{ind} = c_1\left(\bar{Y}_{HT}^{obs}(1,\mathbf{1}) - \bar{Y}_{HT}^{obs}(1,\mathbf{0})\right) + c_2 \left(\bar{Y}_{HT}^{obs}(0,\mathbf{1}) - \bar{Y}_{HT}^{obs}(0,\mathbf{0})\right)
\end{equation}
where the weights $c_1, c_2$ satisfy $c_1 + c_2 = 1$.  
An analogous approach can be used to derive estimators for the ADE and the A$\ell$NNIEs.

Given strong prior knowledge about the values of the potential outcomes, $c_1$ and $c_2$ may be adjusted to reduce the variance of the estimator (similar to Neyman allocation or optimal allocation in \citet{lohr2019sampling}).
However, when no prior knowledge is available, setting $c_1 = c_2 = 1/2$ may be preferable.
Additionally, under this choice of weights, 
estimators can be decomposed analogously to~\eqref{sumofdirectindirectCH3.1} and~\eqref{sumofdirectindirectCH3.2}.
Thus, we proceed assuming $c_1 = c_2 = 1/2$.

\subsection{Estimators Under the no Weak Interaction Assumption}

We begin by proposing unbiased estimators for the average direct, indirect, and $\ell^{th}$ nearest neighbor effects:
\begin{theorem}
Under Assumption~\ref{Assumption2}, the following are unbiased estimators for the ADE, AIE, and the A$\ell$NNIE:
\begin{align} 
    \widehat{\delta^*}_{{HT,dir}} &= 1/2\left(\bar{Y}_{HT}^{obs}(1,\mathbf{1}) - \bar{Y}_{HT}^{obs}(0,\mathbf{1})\right) + 1/2\left(\bar{Y}_{HT}^{obs}(1,\mathbf{0}) - \bar{Y}_{HT}^{obs}(0,\mathbf{0})\right) ,\\
    \widehat{\delta^*}_{{HT,ind}} &= 1/2\left(\bar{Y}_{HT}^{obs}(1,\mathbf{1}) - \bar{Y}_{HT}^{obs}(1,\mathbf{0})\right) + 1/2\left(\bar{Y}_{HT}^{obs}(0,\mathbf{1}) - \bar{Y}_{HT}^{obs}(0,\mathbf{0})\right) , \\
    \widehat{\delta^*}_{HT,\ell} &= 1/2\left(\bar{Y}_{HT}^{obs}(1,\mathbf W^{*}_{\ell}) - \bar{Y}_{HT}^{obs}(1,\mathbf W^{*}_{\ell-1})\right) +1/2\left(\bar{Y}_{HT}^{obs}(0,\mathbf W^{*}_{\ell}) - \bar{Y}_{HT}^{obs}(0,\mathbf W^{*}_{\ell-1})\right).
\end{align} 
\end{theorem}
\begin{proof}
Earlier in this section, we detailed the proof of unbiasedness for   $\widehat{\delta^*}_{{HT,ind}}$. 
The argument for unbiasedness of $\widehat{\delta^*}_{HT,\ell}$ proceeds identically, as does that for $\widehat{\delta^*}_{HT,dir}$ after rewriting~\eqref{noweakint} as
\begin{equation}
    \bar y(1,\mathbf{W}_{\mathcal{N}_{K}})  -
    \bar y(0,\mathbf{W}_{\mathcal{N}_{K}})   = \bar y (1,\mathbf{W}'_{\mathcal{N}_{K}}) - \bar y(0,\mathbf{W}'_{\mathcal{N}_{K}}).
\end{equation}
\end{proof}
 For brevity, to distinguish between these estimators and the estimators defined in Section~\ref{makereference3.5}, we will refer to the estimators $\widehat{\delta}_{{HT,dir}}$, $\widehat{\delta}_{{HT,ind}}$, and, $\widehat{\delta}_{HT,\ell}$ as estimators ``under Assumption 1'' and $\widehat{\delta}^*_{{HT,dir}}$, $\widehat{\delta}^*_{{HT,ind}}$, and, $\widehat{\delta}^*_{HT,\ell}$ as ``under Assumption 2,'' even though both estimators can be used when Assumption~\ref{Assumption2} holds.

These estimators satisfy a similar decomposition to those derived under Assumption~\ref{Assumption1}:
\begin{align}
\label{sumofdirectindirectCH36.1}
\widehat{\delta}_{HT,tot} &= \widehat{\delta^*}_{HT,dir} + \widehat{\delta^*}_{HT,ind}, \\
\label{sumofdirectindirectCH36.2}
\widehat{\delta^*}_{HT,ind} &= \sum_{\ell = 1}^K \widehat{\delta^*}_{HT,\ell} .
\end{align}
Additionally, the variance of $\widehat{\delta^*}_{{HT,dir}}$, $\widehat{\delta^*}_{HT,ind}$, and $\widehat{\delta^*}_{HT,\ell}$ can be derived from repeated application of~\eqref{VAR} and~\eqref{COV}.  
This result generalizes to similar estimators using any four treatment exposures.
\begin{theorem}\label{Theorem6} 
For any four treatment exposures $(W, \mathbf W_{\mathcal{N}_{K}}), (W', \mathbf W'_{\mathcal{N}_{K}}), (W^{*}, \mathbf W^{*}_{\mathcal{N}_{K}}), (W^{*\prime}, \mathbf W^{*\prime}_{\mathcal{N}_{K}})$,
\begin{align}
 \label{varianceass2directCH3}
&\mathbf{Var} \left[1/2\left(\bar{Y}_{HT}^{obs}(W, \mathbf W_{\mathcal{N}_{K}}) - \bar{Y}_{HT}^{obs}(W', \mathbf W'_{\mathcal{N}_{K}})\right) + 1/2\left(\bar{Y}_{HT}^{obs}(W^{*}, \mathbf W^{*}_{\mathcal{N}_{K}}) - \bar{Y}_{HT}^{obs}(W^{*\prime}, \mathbf W^{*\prime}_{\mathcal{N}_{K}})\right)\right] \nn \\  
 = & 1/4\mathbf{Var}(\bar{Y}_{HT}^{obs}(W, \mathbf W_{\mathcal{N}_{K}})) + 1/4\mathbf{Var}(\bar{Y}_{HT}^{obs}(W', \mathbf W'_{\mathcal{N}_{K}})) \nn \\  
 & +  1/4 \mathbf{Var}(\bar{Y}_{HT}^{obs}(W^{*}, \mathbf W^{*}_{\mathcal{N}_{K}})) + 1/4\mathbf{Var}(\bar{Y}_{HT}^{obs}(W^{*'}, \mathbf W^{*'}_{\mathcal{N}_{K}})) \nn \\  
& -  1/2 \mathbf{Cov}(\bar{Y}_{HT}^{obs}(W, \mathbf W_{\mathcal{N}_{K}}), \bar{Y}_{HT}^{obs}(W', \mathbf W'_{\mathcal{N}_{K}})) 
 + 1/2 \mathbf{Cov}(\bar{Y}_{HT}^{obs}(W, \mathbf W_{\mathcal{N}_{K}}), \bar{Y}_{HT}^{obs}(W^{*}, \mathbf W^{*}_{\mathcal{N}_{K}})) \nn \\  
 & -  1/2 \mathbf{Cov}(\bar{Y}_{HT}^{obs}(W, \mathbf W_{\mathcal{N}_{K}}), \bar{Y}_{HT}^{obs}(W^{*'}, \mathbf W^{*'}_{\mathcal{N}_{K}})) 
 - 1/2 \mathbf{Cov}(\bar{Y}_{HT}^{obs}(W', \mathbf W'_{\mathcal{N}_{K}}), \bar{Y}_{HT}^{obs}(W^{*}, \mathbf W^{*}_{\mathcal{N}_{K}}))  \nn \\  
 & +  1/2 \mathbf{Cov}(\bar{Y}_{HT}^{obs}(W', \mathbf W'_{\mathcal{N}_{K}}), \bar{Y}_{HT}^{obs}(W^{*'}, \mathbf W^{*'}_{\mathcal{N}_{K}})) 
 - 1/2 \mathbf{Cov}(\bar{Y}_{HT}^{obs}(W^{*}, \mathbf W^{*}_{\mathcal{N}_{K}}), \bar{Y}_{HT}^{obs}(W^{*'}, \mathbf W^{*'}_{\mathcal{N}_{K}})).
\end{align}
\end{theorem}

\section{Variance Estimators}
\label{makereference3.7}

We extend the work provided in \citep{aronow2013conservative, aronow2017estimating} and \citep{lohr2019sampling} to derive conservative estimators for the variance for all considered estimators under KNNIM.  
This requires estimating all variance and covariance terms in~\eqref{VARDIF} and~\eqref{varianceass2directCH3}.



We begin by estimating the $\mathbf{Var}(\bar{Y}_{HT}^{obs}(W,\mathbf{W}_{\mathcal{N}_{K}}))$ terms.
The standard Horvitz-Thompson estimator for these variances is
\begin{multline}
\widehat{\mathbf{Var}}_{HT}(\bar{Y}_{HT}^{obs}(W,\mathbf{W}_{\mathcal{N}_{K}})) =  \frac{1}{N^2}\sum_{i = 1}^NI_{i}(W,\mathbf{W}_{\mathcal{N}_{K}})[1 - \pi_{i}(W,\mathbf{W}_{\mathcal{N}_{K}})]\left[ \frac{Y_{i}}{\pi_{i}(W,\mathbf{W}_{\mathcal{N}_{K}})}\right]^2
\\  + \frac{1}{N^2}\sum_{i = 1}^N\sum_{j \ne i}I_{i}(W,\mathbf{W}_{\mathcal{N}_{K}})I_{j}(W,\mathbf{W}_{\mathcal{N}_{K}})\frac{[\pi_{ij}(W,\mathbf{W}_{\mathcal{N}_{K}}) - \pi_{i}(W,\mathbf{W}_{\mathcal{N}_{K}})\pi_{j}(W,\mathbf{W}_{\mathcal{N}_{K}})]}{\pi_{ij}(W,\mathbf{W}_{\mathcal{N}_{K}})}
\\ \times \frac{Y_{i}}{\pi_{i}(W,\mathbf{W}_{\mathcal{N}_{K}})}\frac{Y_{j}}{\pi_{j}(W,\mathbf{W}_{\mathcal{N}_{K}})}.
\end{multline}


If the joint probabilities $\pi_{ij}(W,\mathbf{W}_{\mathcal{N}_{K}}) > 0$ for all $i$ and $j$, then this estimated variance is unbiased.
However, under KNNIM, there will likely be incompatibility of this exposure for some $i'$ and $j'$ (see Section~\ref{probcompute} for details), and hence $\pi_{i'j'}(W,\mathbf{W}_{\mathcal{N}_{K}}) = 0$ for these units.
These probabilities lead to bias in the estimate of $\mathbf{Var}(\bar{Y}_{HT}^{obs}(W,\mathbf{W}_{\mathcal{N}_{K}}))$:
\begin{equation}
\label{eq:exvarhat}
\mathbf{E}(\widehat{\mathbf{Var}}_{HT}(\bar{Y}_{HT}^{obs}(W,\mathbf{W}_{\mathcal{N}_{K}}))) =  \mathbf{Var}(\bar{Y}_{HT}^{obs}(W,\mathbf{W}_{\mathcal{N}_{K}})) + A_{Var}
\end{equation}
where 
\begin{equation}
\label{eq:defavar}
A_{Var} = \frac{1}{N^2}\sum_{i = 1}^N\sum_{\substack{j \ne i,\\ \pi_{ij}(W,\mathbf{W}_{\mathcal{N}_{K}}) = 0 }} y_{i}(W,\mathbf{W}_{\mathcal{N}_{K}})y_{j}(W,\mathbf{W}_{\mathcal{N}_{K}}).
\end{equation}

Young’s inequality as derived in \citet{aronow2013conservative, aronow2017estimating} can be used to obtain a quantity $\widehat A_{Var}$ that satisfies $\mathbf{E}(\widehat A_{Var}) \geq |A_{Var}|$.  
We now give a brief proof of Young's inequality as applied to KNNIM, and demonstrate how it can be used to find $\widehat A_{Var}$.
\begin{lemma}
\label{lem:consest}
    For any two (not necessarily unique) exposures $(W,\mathbf{W}_{\mathcal{N}_{K}}), (W',\mathbf{W'}_{\mathcal{N}_{K}})$,
    \begin{align}
 |y_{i}(W,\mathbf{W}_{\mathcal{N}_{K}})y_{j}(W',\mathbf{W'}_{\mathcal{N}_{K}})| \leq 
 \mathbf{E}\left(\frac{I_{i}(W,\mathbf{W}_{\mathcal{N}_{K}})Y_{i}^2}{2\pi_{i}(W,\mathbf{W}_{\mathcal{N}_{K}})} +
 \frac{I_{j}(W',\mathbf{W'}_{\mathcal{N}_{K}})Y_{j}^2}{2\pi_{j}(W',\mathbf{W'}_{\mathcal{N}_{K}})}\right).
    \end{align}
\end{lemma}
\begin{proof}
To see this, first note that
\begin{equation}
   |y_{i}(W,\mathbf{W}_{\mathcal{N}_{K}})y_{j}(W',\mathbf{W'}_{\mathcal{N}_{K}})| = \sqrt{ y^2_{i}(W,\mathbf{W}_{\mathcal{N}_{K}})y^2_{j}(W',\mathbf{W'}_{\mathcal{N}_{K}})}.
\end{equation}
The upper-bound is then proven by applying the arithmetic-mean geometric-mean inequality~\citep{steele2004cauchy}
\begin{equation}
    \sqrt{ y^2_{i}(W,\mathbf{W}_{\mathcal{N}_{K}})y^2_{j}(W',\mathbf{W'}_{\mathcal{N}_{K}})} \leq  \frac{y^2_{i}(W,\mathbf{W}_{\mathcal{N}_{K}})}{2} +
       \frac{y^2_{j}(W',\mathbf{W'}_{\mathcal{N}_{K}})}{2}
\end{equation}
and by noting that the right-hand side can be unbiasedly estimated by
    \begin{equation}
        \frac{I_{i}(W,\mathbf{W}_{\mathcal{N}_{K}})Y_{i}^2}{2\pi_{i}(W,\mathbf{W}_{\mathcal{N}_{K}})} +
 \frac{I_{j}(W',\mathbf{W'}_{\mathcal{N}_{K}})Y_{j}^2}{2\pi_{j}(W',\mathbf{W'}_{\mathcal{N}_{K}})}.
    \end{equation}
\end{proof}

Thus, we define 
\begin{equation}
\widehat{A}_{Var}(W,\mathbf{W}_{\mathcal{N}_{K}}) = \frac{1}{N^2}\sum_{i = 1}^N\sum_{\substack{j \ne i,\\ \pi_{ij}(W,\mathbf{W}_{\mathcal{N}_{K}}) = 0 }}
\left[\frac{I_{i}(W,\mathbf{W}_{\mathcal{N}_{K}})Y_{i}^{2}}{2\pi_{i}(W,\mathbf{W}_{\mathcal{N}_{K}})} + \frac{I_{j}(W,\mathbf{W}_{\mathcal{N}_{K}})Y_{j}^{2}}{2\pi_{j}(W,\mathbf{W}_{\mathcal{N}_{K}})}\right].
\end{equation}
From Lemma~\ref{lem:consest}, linearity of expectations, and the triangle inequality, it follows that
\begin{align}
\mathbf E\left(\widehat{A}_{Var}(W,\mathbf{W}_{\mathcal{N}_{K}})\right) =~ & 
\frac{1}{N^2}\sum_{i = 1}^N\sum_{\substack{j \ne i,\\ \pi_{ij}(W,\mathbf{W}_{\mathcal{N}_{K}}) = 0 }}
\left[\frac{y^2_{i}(W,\mathbf{W}_{\mathcal{N}_{K}})}{2} +
       \frac{y^2_{j}(W,\mathbf{W}_{\mathcal{N}_{K}})}{2}\right] \nn \\ \geq ~ & \frac{1}{N^2}\sum_{i = 1}^N\sum_{\substack{j \ne i,\\ \pi_{ij}(W,\mathbf{W}_{\mathcal{N}_{K}}) = 0 }}
|y_{i}(W,\mathbf{W}_{\mathcal{N}_{K}})y_{j}(W,\mathbf{W}_{\mathcal{N}_{K}})| 
\nn \\ \geq ~ & \left |\frac{1}{N^2}\sum_{i = 1}^N\sum_{\substack{j \ne i,\\ \pi_{ij}(W,\mathbf{W}_{\mathcal{N}_{K}}) = 0 }}
y_{i}(W,\mathbf{W}_{\mathcal{N}_{K}})y_{j}(W,\mathbf{W}_{\mathcal{N}_{K}})\right| = |A_{Var}|. 
\label{eq:consineq}
\end{align}
Therefore, by setting
\begin{equation}
\widehat{\mathbf{Var}}_{A}(\bar{Y}_{HT}^{obs}(W,\mathbf{W}_{\mathcal{N}_{K}})) = \widehat{\mathbf{Var}}_{HT}(\bar{Y}_{HT}^{obs}(W,\mathbf{W}_{\mathcal{N}_{K}})) + \widehat{A}_{Var}(W,\mathbf{W}_{\mathcal{N}_{K}}),
\end{equation}
we obtain a conservative estimate of the variance:
\begin{equation}
    \mathbf{E}\left(\widehat{\mathbf{Var}}_{A}(\bar{Y}_{HT}^{obs}(W,\mathbf{W}_{\mathcal{N}_{K}}))\right) \geq \mathbf{Var}_{HT}(\bar{Y}_{HT}^{obs}(W,\mathbf{W}_{\mathcal{N}_{K}})) + |A_{Var}| + A_{Var} \geq \mathbf{Var}_{HT}(\bar{Y}_{HT}^{obs}(W,\mathbf{W}_{\mathcal{N}_{K}}))
\end{equation}


We apply a similar approach for estimating covariance components. 
Define
\begin{multline}\label{covhat11}
\widehat{\mathbf{Cov}}_{HT}(\bar{Y}_{HT}^{obs}(W,\mathbf{W}_{\mathcal{N}_{K}}),\bar{Y}_{HT}^{obs}(W', \mathbf{W}'_{\mathcal{N}_{K}})) = 
\\ \frac{1}{N^2}\sum_{i = 1}^N\sum_{\substack{j \in \{1,\ldots,N\}\\ \pi_{ij}((W,\mathbf{W}_{\mathcal{N}_{K}}),(W', \mathbf{W}'_{\mathcal{N}_{K}})) > 0 } }
\left[\frac{I_{i}(W,\mathbf{W}_{\mathcal{N}_{K}})I_{j}(W', \mathbf{W}'_{\mathcal{N}_{K}})}{\pi_{ij}((W,\mathbf{W}_{\mathcal{N}_{K}}),(W', \mathbf{W}'_{\mathcal{N}_{K}}))}
\frac{Y_{i}}{\pi_{i}(W,\mathbf{W}_{\mathcal{N}_{K}})}\frac{Y_{j}}{\pi_{j}(W', \mathbf{W}'_{\mathcal{N}_{K}})}\right]
\\ \times \left[\pi_{ij}((W,\mathbf{W}_{\mathcal{N}_{K}}),(W', \mathbf{W}'_{\mathcal{N}_{K}}))-\pi_{i}(W,\mathbf{W}_{\mathcal{N}_{K}})\pi_{j}(W', \mathbf{W}'_{\mathcal{N}_{K}})\right]
\end{multline}

Similar to~\eqref{eq:exvarhat} and~\eqref{eq:defavar}, we have:
\begin{align}
\label{eq:excovhat}
    &\mathbf{E}\left(\widehat{\mathbf{Cov}}_{HT}(\bar{Y}_{HT}^{obs}(W,\mathbf{W}_{\mathcal{N}_{K}}),\bar{Y}_{HT}^{obs}(W', \mathbf{W}'_{\mathcal{N}_{K}})\right) \nn \\=~ & \mathbf{Cov}(\bar{Y}_{HT}^{obs}(W,\mathbf{W}_{\mathcal{N}_{K}}),\bar{Y}_{HT}^{obs}(W', \mathbf{W}'_{\mathcal{N}_{K}})) - A_{Cov}
\end{align}
where 
\begin{equation}
\label{eq:defacov}
A_{Cov} = \sum_{i = 1}^N\sum_{\substack{j \in \{1,\ldots,N\}\\ \pi_{ij}((W,\mathbf{W}_{\mathcal{N}_{K}}),(W', \mathbf{W}'_{\mathcal{N}_{K}})) = 0 }} y_{i}(W,\mathbf{W}_{\mathcal{N}_{K}})y_{j}(W',\mathbf{W}'_{\mathcal{N}_{K}}).
\end{equation}

From repeated application of Lemma~\ref{lem:consest} in a way that mimics the argument in~\eqref{eq:consineq}, we can obtain conservative estimators for the covariance.
Specifically, define:
\begin{multline}\label{eq: COVHATA}
\widehat{\mathbf{Cov}}_{A}(\bar{Y}_{HT}^{obs}(W,\mathbf{W}_{\mathcal{N}_{K}}),\bar{Y}_{HT}^{obs}(W', \mathbf{W}'_{\mathcal{N}_{K}}))= \widehat{\mathbf{Cov}}_{HT}(\bar{Y}_{HT}^{obs}(W,\mathbf{W}_{\mathcal{N}_{K}}),\bar{Y}_{HT}^{obs}(W', \mathbf{W}'_{\mathcal{N}_{K}}))\\ - \frac{1}{N^2}\sum_{i = 1}^N\sum_{\substack{j \in \{1,\ldots,N\}\\ \pi_{ij}((W,\mathbf{W}_{\mathcal{N}_{K}}),(W', \mathbf{W}'_{\mathcal{N}_{K}})) = 0 }}\left[\frac{I_{i}(W,\mathbf{W}_{\mathcal{N}_{K}})Y^2_{i}}{2\pi_{i}(W,\mathbf{W}_{\mathcal{N}_{K}})} + \frac{I_{j}(W', \mathbf{W}'_{\mathcal{N}_{K}})Y^2_{j}}{2\pi_{j}(W', \mathbf{W}'_{\mathcal{N}_{K}})}\right]
\end{multline}
and
\begin{multline}\label{eq: COVHATB}
\widehat{\mathbf{Cov}}_{B}(\bar{Y}_{HT}^{obs}(W,\mathbf{W}_{\mathcal{N}_{K}}),\bar{Y}_{HT}^{obs}(W', \mathbf{W}'_{\mathcal{N}_{K}}))= 
\widehat{\mathbf{Cov}}_{HT}(\bar{Y}_{HT}^{obs}(W,\mathbf{W}_{\mathcal{N}_{K}}),\bar{Y}_{HT}^{obs}(W', \mathbf{W}'_{\mathcal{N}_{K}})) \\  + \frac{1}{N^2}\sum_{i = 1}^N\sum_{\substack{j \in \{1,\ldots,N\}\\ \pi_{ij}((W,\mathbf{W}_{\mathcal{N}_{K}}),(W', \mathbf{W}'_{\mathcal{N}_{K}})) = 0 }}\left[\frac{I_{i}(W,\mathbf{W}_{\mathcal{N}_{K}})Y^2_{i}}{2\pi_{i}(W,\mathbf{W}_{\mathcal{N}_{K}})} + \frac{I_{j}(W', \mathbf{W}'_{\mathcal{N}_{K}})Y^2_{j}}{2\pi_{j}(W', \mathbf{W}'_{\mathcal{N}_{K}})}\right].
\end{multline}
Note, the only difference in these two quantities is that the last term is subtracted in~\eqref{eq: COVHATA} and added in~\eqref{eq: COVHATB}.
The ensuing Lemma immediately follows.
\begin{lemma}
\begin{multline}\label{EXPECTEDCOVHAT2}
  \mathbf{E}(\widehat{\mathbf{Cov}}_{A}(\bar{Y}_{HT}^{obs}(W,\mathbf{W}_{\mathcal{N}_{K}}),\bar{Y}_{HT}^{obs}(W', \mathbf{W}'_{\mathcal{N}_{K}}))) \leq \mathbf{Cov}(\bar{Y}_{HT}^{obs}(W,\mathbf{W}_{\mathcal{N}_{K}}),\bar{Y}_{HT}^{obs}(W', \mathbf{W}'_{\mathcal{N}_{K}})) \\ \leq \mathbf{E}(\widehat{\mathbf{Cov}}_{B}(\bar{Y}_{HT}^{obs}(W,\mathbf{W}_{\mathcal{N}_{K}}),\bar{Y}_{HT}^{obs}(W', \mathbf{W}'_{\mathcal{N}_{K}}))).
\end{multline}
\end{lemma}

Thereby, the conservative variance estimators for the estimators in the previous section under Assumption \ref{Assumption1} for any two exposures $(W,\mathbf{W}_{\mathcal{N}_{K}})$, $(W',\mathbf{W}'_{\mathcal{N}_{K}})$ 
and under Assumption \ref{Assumption2} for any four treatment exposures $(W, \mathbf W_{\mathcal{N}_{K}}), (W', \mathbf W'_{\mathcal{N}_{K}}), (W^{*}, \mathbf W^{*}_{\mathcal{N}_{K}}), (W^{*\prime}, \mathbf W^{*\prime}_{\mathcal{N}_{K}})$ respectively can be derived as follows,
\begin{multline}
\label{VARDIFEstimate}
\widehat{\mathbf{Var}}(\bar{Y}_{HT}^{obs}(W,\mathbf{W}_{\mathcal{N}_{K}}) - \bar{Y}_{HT}^{obs}(W',\mathbf{W'}_{\mathcal{N}_{K}}))=\\
\widehat{\mathbf{Var}}_{A}(\bar{Y}_{HT}^{obs}(W,\mathbf{W}_{\mathcal{N}_{K}})) + \widehat{\mathbf{Var}}_{A}(\bar{Y}_{HT}^{obs}(W',\mathbf{W}'_{\mathcal{N}_{K}})) 
-2\widehat{\mathbf{Cov}}_{A}(\bar{Y}_{HT}^{obs}(W,\mathbf{W}_{\mathcal{N}_{K}}),\bar{Y}_{HT}^{obs}(W',\mathbf{W}'_{\mathcal{N}_{K}})),
\end{multline}

\begin{align}
 \label{varianceass2directCH3Estimate}
&\widehat{\mathbf{Var}} \left[1/2\left(\bar{Y}_{HT}^{obs}(W, \mathbf W_{\mathcal{N}_{K}}) - \bar{Y}_{HT}^{obs}(W', \mathbf W'_{\mathcal{N}_{K}})\right) + 1/2\left(\bar{Y}_{HT}^{obs}(W^{*}, \mathbf W^{*}_{\mathcal{N}_{K}}) - \bar{Y}_{HT}^{obs}(W^{*\prime}, \mathbf W^{*\prime}_{\mathcal{N}_{K}})\right)\right] \nn \\  
 = & 1/4\widehat{\mathbf{Var}}_{A}(\bar{Y}_{HT}^{obs}(W, \mathbf W_{\mathcal{N}_{K}})) + 1/4\widehat{\mathbf{Var}}_{A}(\bar{Y}_{HT}^{obs}(W', \mathbf W'_{\mathcal{N}_{K}})) \nn \\  
 & +  1/4 \widehat{\mathbf{Var}}_{A}(\bar{Y}_{HT}^{obs}(W^{*}, \mathbf W^{*}_{\mathcal{N}_{K}})) + 1/4\widehat{\mathbf{Var}}_{A}(\bar{Y}_{HT}^{obs}(W^{*'}, \mathbf W^{*'}_{\mathcal{N}_{K}})) \nn \\  
& -  1/2 \widehat{\mathbf{Cov}}_{A}(\bar{Y}_{HT}^{obs}(W, \mathbf W_{\mathcal{N}_{K}}), \bar{Y}_{HT}^{obs}(W', \mathbf W'_{\mathcal{N}_{K}})) 
 + 1/2 \widehat{\mathbf{Cov}}_{B}(\bar{Y}_{HT}^{obs}(W, \mathbf W_{\mathcal{N}_{K}}), \bar{Y}_{HT}^{obs}(W^{*}, \mathbf W^{*}_{\mathcal{N}_{K}})) \nn \\  
 & -  1/2 \widehat{\mathbf{Cov}}_{A}(\bar{Y}_{HT}^{obs}(W, \mathbf W_{\mathcal{N}_{K}}), \bar{Y}_{HT}^{obs}(W^{*'}, \mathbf W^{*'}_{\mathcal{N}_{K}})) 
 - 1/2 \widehat{\mathbf{Cov}}_{A}(\bar{Y}_{HT}^{obs}(W', \mathbf W'_{\mathcal{N}_{K}}), \bar{Y}_{HT}^{obs}(W^{*}, \mathbf W^{*}_{\mathcal{N}_{K}}))  \nn \\  
 & +  1/2 \widehat{\mathbf{Cov}}_{B}(\bar{Y}_{HT}^{obs}(W', \mathbf W'_{\mathcal{N}_{K}}), \bar{Y}_{HT}^{obs}(W^{*'}, \mathbf W^{*'}_{\mathcal{N}_{K}})) 
 - 1/2 \widehat{\mathbf{Cov}}_{A}(\bar{Y}_{HT}^{obs}(W^{*}, \mathbf W^{*}_{\mathcal{N}_{K}}), \bar{Y}_{HT}^{obs}(W^{*'}, \mathbf W^{*'}_{\mathcal{N}_{K}})).
\end{align}

This general form of the variance estimators can be applied for all estimators considered under Assumptions \ref{Assumption1} and \ref{Assumption2}. For example, the conservative variance estimators for the A$\ell$NNIE under Assumptions \ref{Assumption1} and \ref{Assumption2} respectively are as follows,

\begin{multline}
\widehat{\mathbf{Var}}(\widehat{\delta}_{{HT,\ell^{th}}}) = 
\widehat{\mathbf{Var}}(\bar{Y}_{HT}^{obs}(0,\mathbf W^{*}_{\ell}) - \bar{Y}_{HT}^{obs}(0,\mathbf W^{*}_{\ell-1})) = \widehat{\mathbf{Var}}_{A}(\bar{Y}_{HT}^{obs}(0, \mathbf W^{*}_{\ell})) + \widehat{\mathbf{Var}}_{A}(\bar{Y}_{HT}^{obs}(0, \mathbf W^{*}_{\ell-1})) 
\\ -2\widehat{\mathbf{Cov}}_{A}(\bar{Y}_{HT}^{obs}(0, \mathbf W^{*}_{\ell}),\bar{Y}_{HT}^{obs}(0, \mathbf W^{*}_{\ell-1})).
\end{multline}

\begin{multline} 
\label{varhatassumption2}
\widehat{\mathbf{Var}}(\widehat{\delta^*}_{{HT,\ell^{th}}}) = 
\widehat{\mathbf{Var}}(1/2\left(\bar{Y}_{HT}^{obs}(1,\mathbf W^{*}_{\ell}) - \bar{Y}_{HT}^{obs}(1,\mathbf W^{*}_{\ell-1})\right) +1/2\left(\bar{Y}_{HT}^{obs}(0,\mathbf W^{*}_{\ell}) - \bar{Y}_{HT}^{obs}(0,\mathbf W^{*}_{\ell-1})\right)) \\ 
= \frac{1}{4}\widehat{\mathbf{Var}}_{A}(\bar{Y}_{HT}^{obs}(1,\mathbf W^{*}_{\ell})) 
+ \frac{1}{4}\widehat{\mathbf{Var}}_{A}(\bar{Y}_{HT}^{obs}(1, \mathbf W^{*}_{\ell-1})) 
\\ + \frac{1}{4}\widehat{\mathbf{Var}}_{A}(\bar{Y}_{HT}^{obs}(0, \mathbf W^{*}_{\ell})) + \frac{1}{4}\widehat{\mathbf{Var}}_{A}(\bar{Y}_{HT}^{obs}(0, \mathbf W^{*}_{\ell-1}))
\\ -\frac{1}{2} \widehat{\mathbf{Cov}}_{A}(\bar{Y}_{HT}^{obs}(1, \mathbf W^{*}_{\ell}), \bar{Y}_{HT}^{obs}(1, \mathbf W^{*}_{\ell-1})) 
+\frac{1}{2} \widehat{\mathbf{Cov}}_{B}(\bar{Y}_{HT}^{obs}(1, \mathbf W^{*}_{\ell}), \bar{Y}_{HT}^{obs}(0, \mathbf W^{*}_{\ell})) 
\\ -\frac{1}{2} \widehat{\mathbf{Cov}}_{A}(\bar{Y}_{HT}^{obs}(1, \mathbf W^{*}_{\ell}), \bar{Y}_{HT}^{obs}(0, \mathbf W^{*}_{\ell-1})) 
-\frac{1}{2}\widehat{\mathbf{Cov}}_{A}(\bar{Y}_{HT}^{obs}(1, \mathbf W^{*}_{\ell-1}), \bar{Y}_{HT}^{obs}(0, \mathbf W^{*}_{\ell})) 
\\ +\frac{1}{2}\widehat{\mathbf{Cov}}_{B}(\bar{Y}_{HT}^{obs}(1, \mathbf W^{*}_{\ell-1}), \bar{Y}_{HT}^{obs}(0, \mathbf W^{*}_{\ell-1})) 
-\frac{1}{2}\widehat{\mathbf{Cov}}_{A}(\bar{Y}_{HT}^{obs}(0, \mathbf W^{*}_{\ell}), \bar{Y}_{HT}^{obs}(0, \mathbf W^{*}_{\ell-1})).
\end{multline}

\section{Simulation}
\label{makereference3.8}

\begin{table}[htbp]

  \centering
  \begin{tabular}[c]{r|ccccccccc}
 \hline
        Model & 1 & 2 & 3 &4 & 5 &  6 &  7 &  8 & 9 \\
\hline
        $\delta_{1}$  & 0 & 0 & 0 & 2 & 2 & 2 & 3 & 3 & 3 \\
        $\delta_{2}$  & 0 & 0 & 0 & 1 & 1 & 1 & 2 & 2 & 2  \\
        $\delta_{3}$  & 0 & 0 & 0 & 0.5 & 0.5 & 0.5 & 1 & 1 & 1  \\
        $\delta_{t}$  & 0 & 1 & 4 & 0 & 1 & 4 & 0 & 1 & 4  \\
\hline
\end{tabular}
  \caption{Interference models.}%
  \label{table3.1}
\end{table}


\begin{table}[t!]
\small
  \caption{Estimates under completely randomized design and Bernoulli randomization model 1.}%
  \label{table3.2}
  \centering
  \begin{tabular}[c]{rc|ccc|ccc}\\
  &&\multicolumn{3}{|c}{Comp. Rand.} & \multicolumn{3}{|c}{Bern. Rand.} \\ \hline
        Estimator & Effects & Emp. EV,  & Emp. Var, & Var Est. & Emp. EV,  & Emp. Var, & Var Est.\\
\hline
       $ \widehat{\delta}_{{HT,tot}}$ & 0 & 0.0779  & 1.1668 &   1.1759 & 0.0059 & 1.2480 & 1.2037 \\
       $\widehat{\delta}_{{HT,dir}}$ & 0 & 0.0337  & 0.6078 &  0.6683  & -0.0150 & 0.6850  & 0.7016 \\
       $ \widehat{\delta^*}_{{HT,dir}}$ & 0 & 0.0303  & 0.2894 &  0.5036 & 0.0103 & 0.3116 & 0.5153 \\
       $\widehat{\delta}_{{HT,ind}}$ & 0 & 0.0442  & 0.9404 &  0.9393 & 0.0210 & 0.8232 & 0.9258 \\
       $ \widehat{\delta^*}_{{HT,ind}}$ & 0 & 0.0476  & 0.5615  & 0.6619 & -0.0044 & 0.5622 & 0.6696\\
       $\widehat{\delta}_{{HT,1^{st}}}$  & 0 & 0.0346  & 0.6273 &  0.6821 & 0.0502 & 0.6502 & 0.6804\\
       $\widehat{\delta^*}_{{HT,1^{st}}}$ & 0 & 0.0191 & 0.2647 &  0.3950 & 0.0173 & 0.2460 & 0.3922\\
        $\widehat{\delta}_{{HT,2^{nd}}}$  & 0 & -0.0111  & 0.4642 &  0.6023 & -0.0181 & 0.4430 & 0.5985 \\
       $\widehat{\delta^*}_{{HT,2^{nd}}}$ & 0 & -0.0021 & 0.2694 &  0.4625 & -0.0085 & 0.2702 & 0.4599\\
        $\widehat{\delta}_{{HT,3^{rd}}}$ & 0 & 0.0207  & 0.4541 &  0.6076 & -0.0110 & 0.4489 & 0.6048 \\
       $\widehat{\delta^*}_{{HT,3^{rd}}}$ & 0 & 0.0306 &   0.3014 &  0.4514 & -0.0132 & 0.3134 & 0.4629 \\
\hline
\end{tabular}
\end{table}

\begin{table}[b!]
\small
  \caption{Estimates under completely randomized design and Bernoulli randomization model 5.}%
  \label{table3.6}
  \centering
  \begin{tabular}[c]{rc|ccc|ccc}\\
  &&\multicolumn{3}{|c}{Comp. Rand.} & \multicolumn{3}{|c}{Bern. Rand.} \\ \hline
        Estimator & Effects & Emp. EV,  & Emp. Var, & Var Est. & Emp. EV,  & Emp. Var, & Var Est.\\
\hline
       $ \widehat{\delta}_{{HT,tot}}$   & 4.5 & 4.5453  & 3.3664 &  3.9196 & 4.6263 & 4.7542 & 5.3978 \\
       $\widehat{\delta}_{{HT,dir}}$  & 1 & 0.9804  & 3.9013 &  4.1313 & 1.1013  & 4.3757 & 4.6502 \\
       $ \widehat{\delta^*}_{{HT,dir}}$ & 1 & 1.0070  & 1.1503 &  1.9781 & 1.0725 & 1.1720 & 2.0765\\
       $\widehat{\delta}_{{HT,ind}}$ & 3.5 & 3.5649  & 1.5750 &  2.3108 & 3.5249  & 1.8276 & 2.5207 \\
       $ \widehat{\delta^*}_{{HT,ind}}$ & 3.5 & 3.5383  & 1.1485 &   1.9574 & 3.5537 & 2.0192 & 2.7832 \\
       $\widehat{\delta}_{{HT,1^{st}}}$  & 2 & 2.0414  & 0.8733 &  1.1102 & 2.0532 & 0.9882 & 1.1739\\
       $\widehat{\delta^*}_{{HT,1^{st}}}$ & 2 & 2.0245  & 0.4873 & 0.9667 & 2.0176 & 0.4648 & 0.9666 \\
        $\widehat{\delta}_{{HT,2^{nd}}}$  & 1 & 1.0003  & 1.4039 &  1.9204 & 0.9878 & 1.4253 & 2.0036 \\
       $\widehat{\delta^*}_{{HT,2^{nd}}}$  & 1 & 0.9930  & 0.9747 &  1.8982 & 0.9582 & 1.1082 & 2.0400 \\
        $\widehat{\delta}_{{HT,3^{rd}}}$ & 0.5 & 0.5230  & 1.8611 &  2.5928 & 0.4838 & 2.2543 & 2.8237 \\
       $\widehat{\delta^*}_{{HT,3^{rd}}}$ & 0.5 & 0.5207 & 1.4523 & 2.3927 & 0.5778 & 2.0343 & 2.7262\\
\hline
\end{tabular}
\end{table}

\begin{table}[tbp]
\small
  \caption{Estimates under completely randomized design and Bernoulli randomization model 9.}%
  \label{table3.10}
   \centering
  \begin{tabular}[c]{rc|ccc|ccc}\\
  &&\multicolumn{3}{|c}{Comp. Rand.} & \multicolumn{3}{|c}{Bern. Rand.} \\ \hline
        Estimator & Effects & Emp. EV,  & Emp. Var, & Var Est. & Emp. EV,  & Emp. Var, & Var Est.\\
\hline
       $ \widehat{\delta}_{{HT,tot}}$   & 10 & 10.0055  & 11.54004 &  14.4106 & 10.2734 & 18.2343 & 21.4468\\
       $\widehat{\delta}_{{HT,dir}}$  & 4 &3.9257  & 14.7035 &  15.0482 & 4.2456 & 17.5899 & 18.4553\\
       $ \widehat{\delta^*}_{{HT,dir}}$ & 4 & 3.9896  & 4.0963 &  7.2216 & 4.1566 & 4.2346 & 7.7429 \\
       $\widehat{\delta}_{{HT,ind}}$ & 6 & 6.0797  & 2.7909 &  4.8503 & 6.0277 & 3.5973 & 5.4461 \\
       $ \widehat{\delta^*}_{{HT,ind}}$ & 6 & 6.0158  & 3.3206 &  6.0066 & 6.1168 & 7.4354 & 9.8815\\
       $\widehat{\delta}_{{HT,1^{st}}}$  & 3 &3.0449  & 1.1582 &  1.6049 & 3.0547 & 1.3709 & 1.7503 \\
       $\widehat{\delta^*}_{{HT,1^{st}}}$ & 3 & 3.0232  & 1.4675 &  3.0152 & 3.0100 & 1.3786 & 3.0134 \\
        $\widehat{\delta}_{{HT,2^{nd}}}$  & 2 &2.0092  & 2.8500 &  3.9616 & 1.9923 & 2.9090 & 4.1353\\
       $\widehat{\delta^*}_{{HT,2^{nd}}}$  & 2 & 1.9838  & 3.2241 &  6.4517 & 1.9150 & 3.8171 & 6.9785 \\
        $\widehat{\delta}_{{HT,3^{rd}}}$ & 1 &1.0256  & 4.5078 &  6.1172 & 0.9806 & 5.3764 & 6.7556\\
       $\widehat{\delta^*}_{{HT,3^{rd}}}$ & 1 & 1.0086  & 5.3193 &  8.5285 & 1.1917 & 7.4434 & 9.8241\\ 
\hline
\end{tabular}
\end{table}

In this section, we assess the performance of our proposed estimators through a simulation study.
In our simulations, we consider a range of indirect effects from no indirect effects to moderate indirect effects and direct effects.
These simulations are designed to verify unbiasedness of our proposed estimators and to assess the tightness of the conservative variance bounds.
Additionally, these simulations will satisfy the assumption of no weak interaction between direct and indirect effects, allowing us to evaluate the gain in precision for estimators exploiting this assumption.

In our simulations, we generate responses under the following KNNIM model with $K = 3$ nearest neighbors:
 \begin{equation} 
Y_{i} = X_{i1} + X_{i2} + X_{i3} + \delta_{1}W_{i1} + \delta_{2}W_{i2} +\delta_{3}W_{i3} + \delta_{t}W_{i} 
\end{equation} 
where $\delta_\ell$ denotes the A$\ell$NNIE, $\ell = 1, 2, 3$, and $\delta_t$ denotes the ADE.  
The $X_{ip}$ are covariates, $X_{ip}\sim N(0,1)$, $p=1,2,3$.  
We use the squared Euclidean distance of the covariates $d(i,j) = \sum_{p=1}^3 (X_{ip} - X_{jp})^2$ as our measure of interaction.
In all models considered, closer neighbors provide stronger indirect effects: $|\delta_{1}|\geq  |\delta_{2} \geq |\delta_{3}|$. 
We assess our estimators under two randomization designs: completely randomized designs (CRD) where half of the $N$ units are assigned to treatment and Bernoulli randomized designs (BRD) with probability $p = 0.5$.

We consider 9 models of interference (see Table~\ref{table3.1}).  The first three models consider no interference, the next three models consider weak interference, and the last three models consider moderate interference.  
Within these categories of indirect effects, we also vary the strength of the direct effect, from no direct effect to a strong direct effect.


For each model, we evaluate the performance of the total, direct, indirect and $\ell_{th}$ nearest neighbor estimators 
under Assumption \ref{Assumption1} and Assumption \ref{Assumption2}. 
The experiment is replicated 1000 times with sample size N = 256.
The marginal and joint probabilities are computed as in Section~\ref{probcompute}.
For each run of 1000 replications, for each estimator, we compute the empirical expected value of the estimates (Emp. EV, the average value of the estimators), empirical variance (Emp. Var, the variance of the estimators),
and the average of the estimated variances (Var Est.). 
For brevity, we only include results for the CRD and BRD designs for Model 1 (Table~\ref{table3.2}), Model 5 (Table~\ref{table3.6}) and Model 9 (Table~\ref{table3.10}).
A complete analysis is included in the supplementary material.

\section{Discussion}
\label{makereference3.9}
  Results indicate, first and foremost, that all estimators under all models accurately estimate all direct and indirect effect parameters---as expected from our theoretical results.
Additionally we also verify the decomposition of estimation of the total effect into direct and indirect components and of the indirect effect estimator into nearest-neighbor components given by~\eqref{sumofdirectindirectCH3.1},~\eqref{sumofdirectindirectCH3.2},~\eqref{sumofdirectindirectCH36.1}, and~\eqref{sumofdirectindirectCH36.2}.
Finally, we confirm that our conservative variance estimators are larger than the true variance, as estimated by the empirical variance.

Under CRD, variance estimates tend to be smaller for closer neighbors under both Assumption \ref{Assumption1} and Assumption \ref{Assumption2}.
The one exception occurs for models that do not exhibit treatment interference. 
Additionally, the variance estimates tend to be smaller for estimators under Assumption~\ref{Assumption2} than under Assumption~\ref{Assumption1} (Tables \ref{table3.2}, \ref{table3.6}).
However, this may not be the case when direct effects are large;
the variance estimates for indirect effects estimators under Assumption 1 are about 50\% smaller than under Assumption \ref{Assumption2} (Table \ref{table3.10}).

Similar results hold for BRD except that Assumption~\ref{Assumption2} did not reduce standard errors of the indirect effects in the moderate direct effect scenario (Table \ref{table3.6}).
Moreover, all variance estimates for all estimators in weak and moderate interference scenarios are slightly larger than those under CRD (Tables \ref{table3.6} and \ref{table3.10}), with differences becoming larger as direct and indirect effects increase.
Hence, CRD may be preferable to BRD for this type of data.

\section{An Analysis of Social Network Experiment}
\label{makereference3.10}

We now apply our estimators to our motivating example which assessed the efficacy of a program designed to promote anti-conflict behaviors in New Jersey middle schools (see Section~\ref{sec:motivexamp}).
Recall that the experiment was explicitly designed to determine whether benefits of the program were propagated through social interactions between students, and moreover, that the strength of the connection between students was explicitly recorded when surveying students in the experiment.  
For our analysis, we use whether or not a student wears an orange wristband as our response.  
A student may be given an orange wristband by a seed student if the original student is observed to be exhibiting anti-conflict behaviors.

To better fit the study into our KNNIM framework, 
we restrict our analysis to seed-eligible students 
who listed two other seed-eligible students among their top 10 closest connections over the prior few weeks. 
This creates a dataset of $N = 348$ students.
We then analyze the data according to a KNNIM model with $K = 2$, where we assume that the seed-student designations were assigned completely at random to the 348 eligible students.
For more discussion of this choice of $K$, please see~\citet{alzubaidi2022detecting}.
Table~\ref{tableDataTTT} gives the number of people assigned to each of the $2^{K+1}= 8$ possible treatment exposures.



We use our proposed estimators to estimate the effect of treatment under Assumption~\ref{Assumption1} and Assumption~\ref{Assumption2}.
Results are presented in Table~\ref{tableDataK2T}. 
Under Assumption \ref{Assumption1},
results suggest that indirect exposure to treatment leads to about a 21\% increase in the probability of wearing a wristband.
The bulk of this effect (about 18\%) is due to the effect from a student's nearest neighbor.  
Our analysis did not find a strong effect of either a direct effect or indirect effects outside of the first nearest neighbor.

Under Assumption \ref{Assumption2}, direct exposure treatment leads to approximately a 5\% increase in the probability of wearing a wrist band while the indirect exposure provides about a 13\% increase in the probability of wearing a wrist band.  
Again, a majority of the indirect effect can be attributed to the treatment status of the nearest neighbor (about 10\%).

Finally, there seems to be differences in the estimates of indirect effect between Assumption~\ref{Assumption1} and Assumption~\ref{Assumption2}.
This could suggest a potential violation of Assumption~\ref{Assumption2} for this dataset; while all estimates of standard errors are considerably smaller under Assumption~\ref{Assumption2}, more care needs to be performed prior to moving forward with such an analysis.
Rigorous methods for testing Assumption~\ref{Assumption2} is an area of future research.

\begin{table}[tbp]
  \caption{Number of units in each exposure of Anti-Conflict Program Experiment with $K =2$, $N= 348$.}%
  \label{tableDataTTT}
  \centering
  \begin{tabular}[c]{lcccc}\\
  \hline
 \multicolumn{5}{c}{Indirect} \\
 \hline
        Direct & $(0,0)$&$(0,1)$&
        $(1,0)$&$(1,1)$ \\
\hline
       Treated & 38 & 42 & 39 & 34 \\
       Control & 40 & 59 & 46 & 50 \\
\hline
\end{tabular}
\end{table}


\begin{table}[tbp]
  \caption{Estimates of Anti-Conflict Program with K = 2 for only Treated School N = 348.}%
  \footnotesize
  \label{tableDataK2T}
  \centering
  \begin{tabular}[c]{rccccccccc}\\
      Estimator & 
      $ \widehat{\delta}_{{HT,tot}}$ &
      $\widehat{\delta}_{{HT,dir}}$ &
      $ \widehat{\delta^*}_{{HT,dir}}$ &
        $\widehat{\delta}_{{HT,ind}}$ &
        $ \widehat{\delta^*}_{{HT,ind}}$ & $\widehat{\delta}_{{HT,1^{st}}}$ & $\widehat{\delta^*}_{{HT,1^{st}}}$ & $\widehat{\delta}_{{HT,2^{nd}}}$ & $\widehat{\delta^*}_{{HT,2^{nd}}}$ \\
\hline
Estimates & 0.1899 & -0.0254 &  0.0559 & 0.2154 &  0.1340 & 0.1788 &  0.1019 & 0.0365 & 0.0320 \\
S.E. &0.0985 &  0.1332 &  0.0863 & 0.0927 & 0.0781 & 0.0822 & 0.0683 &  0.1148 & 0.0934 \\
\hline
\normalsize
\end{tabular}
\end{table}



\section{Conclusion}
\label{makereference3.11}

Substantial effort has been devoted to developing new techniques and models for causal inference when interference is presence---that is, when the effect of treatment is not limited to the unit that receives treatment, but can also impact other units that interact with the original units. 
One recently-proposed model for treatment interference is the $K$-nearest neighbors interference model (KNNIM)~\citep{alzubaidi2022detecting}, in which the treatment status of one unit can affect the response of the $K$ ``closest'' units to that original unit. 

We define causal estimands related to this model---including the average $\ell$th nearest neighbor treatment effect (A$\ell$NNIE)---and show that the indirect effect of treatment can be expressed as a sum of A$\ell$NNIE effects
That is, we show that the KNNIM model is able to determine how the treatment statuses of each of the individual neighbors of a unit contribute to that unit's response.  
We derive unbiased estimators for these estimands and derive conservative variance estimates for these unbiased estimators.
We consider a new assumption on the potential outcomes under treatment interference---no weak interaction between direct and indirect effects---and extend these estimators under this assumption.
We perform a simulation study to determine the efficacy of these estimators under various settings.

We apply our methodology to a recent experiment testing the efficacy of a program designed to reduce conflict in middle schools in New Jersey~\citep{paluck2016changing}.  
This experiment has two appealing characteristics: it was designed to propagate through peer-to-peer interactions between students and students were explicitly surveyed about which students were their ``closest connections.'' 
Using our approach, we give evidence that that one measure in particular---wearing a wristband given for demonstrating anti-conflict behaviors---is affected primarily by the students' closest connection.

\bibliographystyle{chicago}
\bibliography{References}

\newpage
\begin{center}
{\large\bf SUPPLEMENTARY MATERIAL}
\end{center}


\appendix

\section{Additional simulation results}
\label{AppendixC}

We provide complete simulation results for all models of response given in Table 1.

\subsection{Estimates under complete randomization}
\begin{table}[h!]
  \caption{Estimates Under Completely Randomized Design Model 2}%
  \label{table3.3}
  \centering
  \begin{tabular}[c]{lccccc}\\
        Estimator & Effects & Emp.Estimates & Emp.Var & Emp.S.D. & Var Estimate \\
\hline
       $ \widehat{\delta}_{{HT,tot}}$ & 1 & 1.0707  & 1.306 & 1.1430 & 1.3314 \\
       $ \widehat{\delta^*}_{{HT,tot}}$ & 1 & 1.0707  & 1.306 & 1.1430 & 1.3314 \\
       $\widehat{\delta}_{{HT,dir}}$ & 1 & 1.0265  & 0.7273 & 0.8528 & 0.8050 \\
       $ \widehat{\delta^*}_{{HT,dir}}$ & 1 & 1.0300  & 0.3377 & 0.5812 &  0.6214\\
       $\widehat{\delta}_{{HT,ind}}$ & 0 & 0.0442  & 0.9404 & 0.9697 & 0.9393\\
       $ \widehat{\delta^*}_{{HT,ind}}$ & 0 & 0.0406  &  0.6081 & 0.7798 & 0.7381\\
       $\widehat{\delta}_{{HT,1^{st}}}$  & 0 & 0.0346  & 0.6273 & 0.7920 & 0.6821 \\
       $\widehat{\delta^*}_{{HT,1^{st}}}$ & 0 & 0.01758  & 0.3009 & 0.5486 & 0.4749  \\
        $\widehat{\delta}_{{HT,2^{nd}}}$ & 0 & -0.0111   & 0.4642 & 0.6813 & 0.6023 \\
       $\widehat{\delta^*}_{{HT,2^{nd}}}$  & 0 & -0.0052  & 0.3328 & 0.5769 & 0.5601 \\
        $\widehat{\delta}_{{HT,3^{rd}}}$ & 0 & 0.0207  & 0.4541 & 0.6738 & 0.6076 \\
       $\widehat{\delta^*}_{{HT,3^{rd}}}$ & 0 & 0.0282 & 0.3525 & 0.5937 & 0.5359\\
\hline
\end{tabular}
\end{table}

\begin{table}
  \caption{Estimates Under Completely Randomized Design Model 3}%
  \label{table3.4}
  \centering
  \begin{tabular}[c]{lccccc}\\
        Estimator & Effects & Emp.Estimates & Emp.Var & Emp.S.D. & Var Estimate \\
\hline
       $ \widehat{\delta}_{{HT,tot}}$   & 4 &  4.0489  & 2.9225 & 1.7095 & 3.3552\\
       $ \widehat{\delta^*}_{{HT,tot}}$  & 4 &  4.0489 & 2.9225 & 1.7095 & 3.3552\\
       $\widehat{\delta}_{{HT,dir}}$  & 4 & 4.0047  & 2.2827 & 1.5108 & 2.5258  \\
       $ \widehat{\delta^*}_{{HT,dir}}$   & 4 & 4.0292  & 0.9267 & 0.9626 & 2.1155 \\
       $\widehat{\delta}_{{HT,ind}}$ & 4 & 0.0442  &  0.9404 & 0.9697 & 0.9393 \\
       $ \widehat{\delta^*}_{{HT,ind}}$   & 0 & 0.01978  & 1.1883 & 1.0901 & 1.7146 \\
       $\widehat{\delta}_{{HT,1^{st}}}$  & 0 &  0.0346  & 0.6273 & 0.7920 & 0.6821 \\
       $\widehat{\delta^*}_{{HT,1^{st}}}$ & 0 & 0.01286  &  0.8246 &  0.9081 & 1.4569 \\
        $\widehat{\delta}_{{HT,2^{nd}}}$ & 0 & -0.0111   & 0.4642 & 0.6813 & 0.6023 \\
       $\widehat{\delta^*}_{{HT,2^{nd}}}$  & 0 & -0.01445  & 1.0133 & 1.0066 & 1.7389 \\
        $\widehat{\delta}_{{HT,3^{rd}}}$ & 0 & 0.0207  & 0.4541 & 0.6738 & 0.6076 \\
       $\widehat{\delta^*}_{{HT,3^{rd}}}$ & 0 & 0.0213  & 1.0711 & 1.0349 & 1.6299 \\
\hline
\end{tabular}
\end{table}

\begin{table}
  \caption{Estimates Under Completely Randomized Design Model 4}%
  \label{table3.5}
  \centering
  \begin{tabular}[c]{lccccc}\\
        Estimator & Effects & Emp.Estimates & Emp.Var & Emp.S.D. & Var Estimate \\
\hline
       $ \widehat{\delta}_{{HT,tot}}$ & 3.5 & 3.5526  & 2.5285 & 1.5901 &  2.8557\\
       $ \widehat{\delta^*}_{{HT,tot}}$ & 3.5 & 3.5526  & 2.5285 & 1.5901 &  2.8557\\
       $\widehat{\delta}_{{HT,dir}}$ & 0 & -0.0123  & 2.9637 & 1.7215 & 3.2323\\
       $ \widehat{\delta^*}_{{HT,dir}}$ & 0 & 0.0073  & 0.8913 & 0.9441 & 1.5485\\
       $\widehat{\delta}_{{HT,ind}}$ & 3.5 & 3.5649  & 1.5750 & 1.2550 & 2.3108 \\
       $ \widehat{\delta^*}_{{HT,ind}}$  & 3.5 & 3.5453  & 0.9523 & 0.9758 & 1.6828 \\
       $\widehat{\delta}_{{HT,1^{st}}}$  & 2 & 2.0414  & 0.8733 & 0.9345 & 1.1102 \\
       $\widehat{\delta^*}_{{HT,1^{st}}}$ & 2 & 2.0261  & 0.3901 &  0.6245 & 0.7582\\
        $\widehat{\delta}_{{HT,2^{nd}}}$  & 1 & 1.0003  & 1.4039 & 1.1848 & 1.9204\\
       $\widehat{\delta^*}_{{HT,2^{nd}}}$  & 1 & 0.9961  & 0.7423 & 0.8615 & 1.4553 \\
        $\widehat{\delta}_{{HT,3^{rd}}}$ & 0.5 & 0.5230  & 1.8611 & 1.3642 &  2.5928 \\
       $\widehat{\delta^*}_{{HT,3^{rd}}}$ & 0.5 & 0.5230 & 1.1157 & 1.0563 & 1.8726\\
\hline
\end{tabular}
\end{table}

\begin{table}
  \caption{Estimates Under Completely Randomized Design Model 6}%
  \label{table3.7}
  \centering
  \begin{tabular}[c]{lccccc}\\
        Estimator & Effects & Emp.Estimates  & Emp.Var & Emp.S.D. & Var Estimate \\
\hline
       $ \widehat{\delta}_{{HT,tot}}$   & 7.5 & 7.5236  & 7.0767 & 2.6602 & 8.6687 \\
       $ \widehat{\delta^*}_{{HT,tot}}$  & 7.5 & 7.5236  & 7.0767 & 2.6602 & 8.6687 \\
       $\widehat{\delta}_{{HT,dir}}$  & 4 & 3.9586  & 7.9109 & 2.8126 & 8.1389\\
       $ \widehat{\delta^*}_{{HT,dir}}$ & 4 & 4.0061  & 2.3713 & 1.5399 & 4.4075\\
       $\widehat{\delta}_{{HT,ind}}$ & 3.5 & 3.5649  & 1.5750 & 1.2550 & 2.3108\\
       $ \widehat{\delta^*}_{{HT,ind}}$ & 3.5 & 3.5174  & 2.1777 & 1.4757 & 3.5289 \\
       $\widehat{\delta}_{{HT,1^{st}}}$  & 2 & 2.0414  & 0.8733 & 0.9345 & 1.1102 \\
       $\widehat{\delta^*}_{{HT,1^{st}}}$ & 2 & 2.0198  & 1.1942 & 1.0928 & 2.3348\\
        $\widehat{\delta}_{{HT,2^{nd}}}$  & 1 & 1.0003  & 1.4039 & 1.1848 & 1.9204\\
       $\widehat{\delta^*}_{{HT,2^{nd}}}$  & 1 & 0.9838  & 2.1622 & 1.4704 & 4.1131\\
        $\widehat{\delta}_{{HT,3^{rd}}}$ & 0.5 & 0.5230  & 1.8611 & 1.3642 &  2.5928\\
       $\widehat{\delta^*}_{{HT,3^{rd}}}$ & 0.5 & 0.5138  & 3.0275 & 1.7399 & 4.7936 \\
\hline
\end{tabular}
\end{table}

\begin{table}
  \caption{Estimates Under Completely Randomized Design Model 7}%
  \label{table3.8}
  \centering
  \begin{tabular}[c]{lccccc}\\
        Estimator & Effects & Emp.Estimates & Emp.Var & Emp.S.D. & Var Estimate \\
\hline
       $ \widehat{\delta}_{{HT,tot}}$ & 6 & 6.0344 & 4.9972 & 2.2354 & 6.0021\\
       $ \widehat{\delta^*}_{{HT,tot}}$& 6 & 6.0344 & 4.9972 & 2.2354 & 6.0021 \\
       $\widehat{\delta}_{{HT,dir}}$ & 0 & -0.0452 & 7.4190 & 2.7237 & 7.9636 \\
       $ \widehat{\delta^*}_{{HT,dir}}$ & 0 & -0.0091 & 2.0143 & 1.4192 & 3.4718\\
       $\widehat{\delta}_{{HT,ind}}$ & 6 & 6.0797 & 2.7909 & 1.6705 & 4.850\\
       $ \widehat{\delta^*}_{{HT,ind}}$ & 6 & 6.0436 & 1.6675 & 1.2913 & 3.5938\\
       $\widehat{\delta}_{{HT,1^{st}}}$  & 3 &3.0449  & 1.1582 & 1.0762 & 1.6049\\
       $\widehat{\delta^*}_{{HT,1^{st}}}$ & 3 & 3.0295  & 0.5412 & 0.7357 & 1.1811\\
        $\widehat{\delta}_{{HT,2^{nd}}}$  & 2 &2.0092 & 2.8500 & 1.6882 & 3.9616\\
       $\widehat{\delta^*}_{{HT,2^{nd}}}$  & 2 & 1.9961  & 1.4097 & 1.1873 & 2.9671\\
        $\widehat{\delta}_{{HT,3^{rd}}}$ & 1 &1.0256  & 4.5078 & 2.1231 & 6.1172\\
       $\widehat{\delta^*}_{{HT,3^{rd}}}$ & 1 & 1.0179  & 2.6139 & 1.6167 & 4.4043\\
\hline
\end{tabular}
\end{table}

\begin{table}
  \caption{Estimates Under Completely Randomized Design Model 8}%
  \label{table3.9}
  \centering
  \begin{tabular}[c]{lccccc}\\
        Estimator & Effects & Emp.Estimates & Emp.Var & Emp.S.D. & Var Estimate \\
\hline
       $ \widehat{\delta}_{{HT,tot}}$   & 7 & 7.0272  & 6.3337 & 2.5166 & 7.7149 \\
       $ \widehat{\delta^*}_{{HT,tot}}$  & 7 & 7.0272  & 6.3337 & 2.5166 & 7.7149 \\
       $\widehat{\delta}_{{HT,dir}}$  & 1 & 0.9474  & 8.9409 & 2.9901 & 9.4071 \\
       $ \widehat{\delta^*}_{{HT,dir}}$   & 1 & 0.9905 & 2.4239 & 1.5568 & 4.1241 \\
       $\widehat{\delta}_{{HT,ind}}$ & 6 & 6.0797 & 2.7909 & 1.6705 & 4.8503\\
       $ \widehat{\delta^*}_{{HT,ind}}$ & 6 & 6.0366  & 1.9707 & 1.4038 & 4.0101\\
       $\widehat{\delta}_{{HT,1^{st}}}$  & 3 & 3.0449  & 1.1582 & 1.0762 & 1.6049\\
       $\widehat{\delta^*}_{{HT,1^{st}}}$ & 3 & 3.0280 & 0.6690 & 0.8179 & 1.4541\\
        $\widehat{\delta}_{{HT,2^{nd}}}$  & 2 & 2.0092 & 2.8500 & 1.6882 & 3.9616 \\
       $\widehat{\delta^*}_{{HT,2^{nd}}}$  & 2 & 1.9930 & 1.7408 & 1.3194 & 3.6167\\
        $\widehat{\delta}_{{HT,3^{rd}}}$ & 1 & 1.0256  & 4.5078 & 2.1231 & 6.1172\\
       $\widehat{\delta^*}_{{HT,3^{rd}}}$ & 1 & 1.0156 & 3.1489 & 1.7745 & 5.2252 \\
\hline
\end{tabular}
\end{table}
\newpage


\subsection{Estimates under Bernoulli randomization}

\begin{table}[h!]
  \caption{Estimates Under Bernoulli Randomization Model 2}%
  \label{table2Ber}
  \centering
  \begin{tabular}[c]{lccccc}\\
        Estimator & Effects & Emp.Estimates  & Emp.Var & Emp.S.D. & Var Estimate \\
\hline
        $ \widehat{\delta}_{{HT,tot}}$   & 1 & 1.0326  & 1.4420 & 1.2008 & 1.4406 \\
       $ \widehat{\delta^*}_{{HT,tot}}$   & 1 & 1.0326 & 1.4420 & 1.2008 & 1.4406 \\
       $\widehat{\delta}_{{HT,dir}}$  & 1 & 1.0116  & 0.8786 & 0.9373 & 0.8960 \\
       $ \widehat{\delta^*}_{{HT,dir}}$   & 1 & 1.0277 & 0.3573 & 0.5978 & 0.6357  \\
       $\widehat{\delta}_{{HT,ind}}$ & 0 & 0.0210  & 0.8232 & 0.9073 & 0.9258 \\
       $ \widehat{\delta^*}_{{HT,ind}}$   & 0 & 0.0049  & 0.6462 & 0.8038 & 0.7994  \\
       $\widehat{\delta}_{{HT,1^{st}}}$  & 0 & 0.0502  & 0.6502 & 0.8063 & 0.6804 \\
       $\widehat{\delta^*}_{{HT,1^{st}}}$ & 0 & 0.0142 & 0.2858 & 0.5346 & 0.4783  \\
        $\widehat{\delta}_{{HT,2^{nd}}}$  & 0 & -0.0181  & 0.4430 & 0.6656 & 0.5985 \\
       $\widehat{\delta^*}_{{HT,2^{nd}}}$  & 0 & -0.0177  & 0.3250 & 0.5700 & 0.5625  \\
        $\widehat{\delta}_{{HT,3^{rd}}}$ & 0 & -0.0110  & 0.4489 & 0.6700 & 0.6048 \\
       $\widehat{\delta^*}_{{HT,3^{rd}}}$ & 0 & 0.0084  & 0.3846 & 0.6202 & 0.5497  \\
\hline
\end{tabular}
\end{table}

\begin{table}[tbp]
  \caption{Estimates Under Bernoulli Randomization Model 3}%
  \label{table3Ber}
  \centering
  \begin{tabular}[c]{lccccc}\\
        Estimator & Effects & Emp.Estimates  & Emp.Var & Emp.S.D. & Var Estimate \\
\hline
        $ \widehat{\delta}_{{HT,tot}}$   & 4 & 4.1129  & 4.0302 & 2.0075 & 4.5346  \\
       $ \widehat{\delta^*}_{{HT,tot}}$  & 4 & 4.1129  & 4.0302 & 2.0075 & 4.5346  \\
       $\widehat{\delta}_{{HT,dir}}$  & 4 & 4.0919 & 3.4654 & 1.8615 & 3.6069  \\
       $ \widehat{\delta^*}_{{HT,dir}}$   & 4 & 4.0797 & 0.9582 & 0.9789 & 2.2095  \\
       $\widehat{\delta}_{{HT,ind}}$ & 0 & 0.0210  & 0.8232 & 0.9073 & 0.9258 \\
       $ \widehat{\delta^*}_{{HT,ind}}$   & 0 & 0.0331 & 1.8081 & 1.3446 & 2.3816  \\
       $\widehat{\delta}_{{HT,1^{st}}}$  & 0 & 0.0502  & 0.6502 & 0.8063 & 0.6804 \\
       $\widehat{\delta^*}_{{HT,1^{st}}}$ & 0 & 0.0051 & 0.8466 & 0.9201 & 1.5260  \\
        $\widehat{\delta}_{{HT,2^{nd}}}$  & 0 & -0.0181 & 0.4430 & 0.6656 & 0.5985 \\
       $\widehat{\delta^*}_{{HT,2^{nd}}}$  & 0 & -0.0456 & 1.0226 & 1.0112 & 1.7932  \\
        $\widehat{\delta}_{{HT,3^{rd}}}$ & 0 & -0.0110  & 0.4489 & 0.6700 & 0.6048 \\
       $\widehat{\delta^*}_{{HT,3^{rd}}}$ & 0 & 0.0737  & 1.2557 & 1.1205 & 1.7099  \\
\hline
\end{tabular}
\end{table}

\begin{table}[tbp]
  \caption{Estimates Under Bernoulli Randomization Model 4}%
  \label{table4Ber}
  \centering
  \begin{tabular}[c]{lccccc}\\
        Estimator & Effects & Emp.Estimates & Emp.Var & Emp.S.D. & Var Estimate \\
\hline
        $ \widehat{\delta}_{{HT,tot}}$  & 3.5 & 3.5995  & 3.3899 & 1.8411 & 3.7707  \\
       $ \widehat{\delta^*}_{{HT,tot}}$  & 3.5 & 3.5995 & 3.3899 & 1.8411 & 3.7707  \\
       $\widehat{\delta}_{{HT,dir}}$  & 0 & 0.0745 & 3.1656 & 1.7792 & 3.4406  \\
       $ \widehat{\delta^*}_{{HT,dir}}$  & 0 & 0.0552 & 0.9018 & 0.9496 & 1.6046  \\
       $\widehat{\delta}_{{HT,ind}}$ & 3.5 & 3.5249  & 1.8276 & 1.3519 & 2.5207  \\
       $ \widehat{\delta^*}_{{HT,ind}}$   & 3.5 & 3.5443  & 1.5038 & 1.2262 & 2.1947  \\
       $\widehat{\delta}_{{HT,1^{st}}}$  & 2 & 2.0532  & 0.9882 & 0.9940 & 1.1739  \\
       $\widehat{\delta^*}_{{HT,1^{st}}}$ & 2 & 2.0206  & 0.3871 & 0.6222 & 0.7689  \\
        $\widehat{\delta}_{{HT,2^{nd}}}$  & 1 & 0.9878  & 1.4253 & 1.1938 & 2.0036  \\
       $\widehat{\delta^*}_{{HT,2^{nd}}}$  & 1 & 0.9675 & 0.8350 & 0.9137 & 1.5577  \\
        $\widehat{\delta}_{{HT,3^{rd}}}$ & 0.5 & 0.4838 & 2.2543 & 1.5014 & 2.8237\\
       $\widehat{\delta^*}_{{HT,3^{rd}}}$ & 0.5 & 0.5561  & 1.5671 & 1.2518 & 2.1355  \\
\hline
\end{tabular}
\end{table}

\begin{table}[tbp]
  \caption{Estimates Under Bernoulli Randomization Model 6}%
  \label{table6Ber}
  \centering
  \begin{tabular}[c]{lccccc}\\
        Estimator & Effects & Emp.Estimates & Emp.Var & Emp.S.D. & Var Estimate \\
\hline
        $ \widehat{\delta}_{{HT,tot}}$   & 7.5 & 7.7065 & 10.8531 & 3.2944 & 12.6623  \\
       $ \widehat{\delta^*}_{{HT,tot}}$  & 7.5 & 7.7065 & 10.8531 & 3.2944 & 12.6623  \\
       $\widehat{\delta}_{{HT,dir}}$  & 4 & 4.1815 & 10.0122 & 3.1642 & 10.4070 \\
       $ \widehat{\delta^*}_{{HT,dir}}$  & 4 & 4.1245 & 2.4463 & 1.5640 & 4.7047  \\
       $\widehat{\delta}_{{HT,ind}}$ & 3.5 & 3.5249 & 1.8276 & 1.3519 & 2.5207  \\
       $ \widehat{\delta^*}_{{HT,ind}}$   & 3.5 & 3.5819 & 4.4755 & 2.1155 & 5.7415  \\
       $\widehat{\delta}_{{HT,1^{st}}}$  & 2 & 2.0532 & 0.9882 & 0.9940 & 1.1739  \\
       $\widehat{\delta^*}_{{HT,1^{st}}}$ & 2 & 2.0084  & 1.1391 & 1.0672 & 2.3490  \\
        $\widehat{\delta}_{{HT,2^{nd}}}$  & 1 & 0.9878 & 1.4253 & 1.1938 & 2.0036 \\
       $\widehat{\delta^*}_{{HT,2^{nd}}}$  & 1 & 0.9304 & 2.4611 & 1.5687 & 4.4098 \\
        $\widehat{\delta}_{{HT,3^{rd}}}$ & 0.5 & 0.4838  & 2.2543 & 1.5014 & 2.8237 \\
       $\widehat{\delta^*}_{{HT,3^{rd}}}$ & 0.5 & 0.6431  & 4.0935 & 2.0232 & 5.3983  \\
\hline
\end{tabular}
\end{table}

\begin{table}[tbp]
  \caption{Estimates Under Bernoulli Randomization Model 7}%
  \label{table7Ber}
  \centering
  \begin{tabular}[c]{lccccc}\\
        Estimator & Effects & Emp.Estimates  & Emp.Var & Emp.S.D. & Var Estimate \\
\hline
        $ \widehat{\delta}_{{HT,tot}}$   & 6 & 6.1664 & 7.4275 & 2.7253 & 8.5832  \\
       $ \widehat{\delta^*}_{{HT,tot}}$  & 6 & 6.1664 & 7.4275 & 2.7253 & 8.5832  \\
       $\widehat{\delta}_{{HT,dir}}$  & 0 & 0.1386  & 7.8388 & 2.7997 & 8.5881  \\
       $ \widehat{\delta^*}_{{HT,dir}}$   & 0 & 0.0872 & 2.0487 & 1.4313 & 3.6387  \\
       $\widehat{\delta}_{{HT,ind}}$ & 6 & 6.0277 & 3.5973 & 1.8966 & 5.4461  \\
       $ \widehat{\delta^*}_{{HT,ind}}$   & 6 & 6.0791  & 3.2309 & 1.7974 & 5.0241  \\
       $\widehat{\delta}_{{HT,1^{st}}}$ & 3 & 3.0547  & 1.3709 & 1.1708 & 1.7503  \\
       $\widehat{\delta^*}_{{HT,1^{st}}}$ & 3 & 3.0223  & 0.5510 & 0.7423 & 1.2101  \\
        $\widehat{\delta}_{{HT,2^{nd}}}$  & 2 & 1.9923 & 2.9090 & 1.7056 & 4.1353 \\
       $\widehat{\delta^*}_{{HT,2^{nd}}}$  & 2 & 1.9521 & 1.6716 & 1.2929 & 3.2181  \\
        $\widehat{\delta}_{{HT,3^{rd}}}$ & 1 & 0.9806  & 5.3764 & 2.3187 & 6.7556  \\
       $\widehat{\delta^*}_{{HT,3^{rd}}}$ & 1 & 1.1047 & 3.8051 & 1.9506 & 5.1576 \\
\hline
\end{tabular}
\end{table}

\begin{table}[tbp]
  \caption{Estimates Under Bernoulli Randomization Model 8}%
  \label{table8Ber}
  \centering
  \begin{tabular}[c]{lccccc}\\
        Estimator & Effects & Emp.Estimates & Emp.Var & Emp.S.D. & Var Estimate \\
\hline
        $ \widehat{\delta}_{{HT,tot}}$   & 7 & 7.1931 & 9.6277 & 3.1028 & 11.2033  \\
       $ \widehat{\delta^*}_{{HT,tot}}$  & 7 & 7.1931 & 9.6277 & 3.1028 & 11.2033  \\
       $\widehat{\delta}_{{HT,dir}}$ & 1 & 1.1654 & 9.7750 & 3.1265 & 10.5230  \\
       $ \widehat{\delta^*}_{{HT,dir}}$   & 1 & 1.1045 & 2.4792 & 1.5745 & 4.3616  \\
       $\widehat{\delta}_{{HT,ind}}$ & 6 & 6.0277 & 3.5973 & 1.8966 & 5.4461  \\
       $ \widehat{\delta^*}_{{HT,ind}}$   & 6 & 6.0885 & 4.0546 & 2.0136 & 5.9402 \\
       $\widehat{\delta}_{{HT,1^{st}}}$ & 3 & 3.0547 & 1.3709 & 1.1708 & 1.7503  \\
       $\widehat{\delta^*}_{{HT,1^{st}}}$ & 3 & 3.0192 & 0.6476 & 0.8047 & 1.4636  \\
        $\widehat{\delta}_{{HT,2^{nd}}}$ & 2 & 1.9923 & 2.9090 & 1.7056 & 4.1353  \\
       $\widehat{\delta^*}_{{HT,2^{nd}}}$  & 2 & 1.9428 & 2.0746 & 1.4403 & 3.9275  \\
        $\widehat{\delta}_{{HT,3^{rd}}}$ & 1 & 0.9806 & 5.3764 & 2.3187 & 6.7556  \\
       $\widehat{\delta^*}_{{HT,3^{rd}}}$ & 1 & 1.1264 & 4.5503 & 2.1331 & 6.0993 \\
\hline
\end{tabular}
\end{table}

\end{document}